\documentclass[aps,onecolumn,nofootinbib,tightenlines,superscriptaddress]{revtex4}

\usepackage{amssymb,amstext,amsmath,amsthm,slashed}
\usepackage[dvips]{graphicx}
\usepackage{latexsym}
\usepackage{psfrag}
\usepackage{amsfonts}
\usepackage{bbm} 
\usepackage{color}
\usepackage{verbatim}
\usepackage{bbm}
\usepackage{dcolumn}
\usepackage{tabularx}
\usepackage{leftidx}
\usepackage{tensor}
\usepackage{boxedminipage}
\usepackage{fancyvrb}
\usepackage{float}
\usepackage{wrapfig}
\usepackage{booktabs}
\usepackage{dsfont} 
\usepackage{mathrsfs}
\usepackage[hidelinks]{hyperref}

\newtheorem{Def}{Definition}[section]
\newtheorem{Thm}[Def]{Theorem}
\newtheorem{Prp}[Def]{Proposition}
\newtheorem{Lemma}[Def]{Lemma}

\newcommand{\Proof}{\begin{proof}}
\newcommand{\QED}{\end{proof} \noindent}
\newcommand{\Sl}{\mbox{$\prec \!\!$ \nolinebreak}}
\newcommand{\Sr}{\mbox{\nolinebreak $\succ$}}

\newcommand{\1}{\mbox{\rm 1 \hspace{-1.05 em} 1}}

\DeclareFontFamily{U}{mathx}{\hyphenchar\font45}
\DeclareFontShape{U}{mathx}{m}{n}{
	<5> <6> <7> <8> <9> <10>
	<10.95> <12> <14.4> <17.28> <20.74> <24.88>
	mathx10
}{}
\DeclareSymbolFont{mathx}{U}{mathx}{m}{n}
\DeclareFontSubstitution{U}{mathx}{m}{n}
\DeclareMathAccent{\widecheck}{0}{mathx}{"71}
\DeclareMathAccent{\wideparen}{0}{mathx}{"75}

\begin{document}

\title{An Integral Spectral Representation of the Massive Dirac Propagator \\ in the Kerr Geometry in Eddington--Finkelstein-type Coordinates \vspace{0.1cm}}

\author{Felix Finster\vspace{0.1cm}\let\thefootnote\relax\footnote{e-mail: felix.finster@mathematik.uni-regensburg.de}}

\affiliation{Universit\"at Regensburg, Fakult\"at f\"ur Mathematik, 93040 Regensburg, Germany \vspace{0.1cm}} 

\author{Christian R\"oken\let\thefootnote\relax\footnote{e-mail: christian.roeken@mathematik.uni-regensburg.de}}
 
\affiliation{Universit\"at Regensburg, Fakult\"at f\"ur Mathematik, 93040 Regensburg, Germany \vspace{0.1cm}} 

\affiliation{Departamento de Geometr\'ia y Topolog\'ia, Facultad de Ciencias - Universidad de Granada, Campus de Fuentenueva s/n, 18071 Granada, Spain \vspace{0.4cm}}

\date{August 2018}

\begin{abstract}
\vspace{0.4cm} \noindent \textbf{\footnotesize ABSTRACT.} \, We consider the massive Dirac equation in the non-extreme Kerr geometry in horizon-penetrating advanced Eddington--Finkelstein-type coordinates and derive a functional analytic integral representation of the associated propagator using the spectral theorem for unbounded self-adjoint operators, Stone's formula, and quantities arising in the analysis of Chandrasekhar's separation of variables. This integral representation describes the dynamics of Dirac particles outside and across the event horizon, up to the Cauchy horizon. In the derivation, we first write the Dirac equation in Hamiltonian form and show the essential self-adjointness of the Hamiltonian. For the latter purpose, as the Dirac Hamiltonian fails to be elliptic at the event and the Cauchy horizon, we cannot use standard elliptic methods of proof. Instead, we employ a new, general method for mixed initial-boundary value problems that combines results from the theory of symmetric hyperbolic systems with near-boundary elliptic methods. In this regard and since the time evolution may not be unitary because of Dirac particles impinging on the ring singularity, we also impose a suitable Dirichlet-type boundary condition on a time-like inner hypersurface placed inside the Cauchy horizon, which has no effect on the dynamics outside the Cauchy horizon. We then compute the resolvent of the Dirac Hamiltonian via the projector onto a finite-dimensional, invariant spectral eigenspace of the angular operator and the radial Green's matrix stemming from Chandrasekhar's separation of variables. Applying Stone's formula to the spectral measure of the Hamiltonian in the spectral decomposition of the Dirac propagator, that is, by expressing the spectral measure in terms of this resolvent, we obtain an explicit integral representation of the propagator.
\end{abstract}

\setcounter{tocdepth}{2}

\vspace{0.1cm}

\maketitle

\tableofcontents

%-----------------------------------------------------
\section{Introduction}
%-----------------------------------------------------

\noindent In \cite{FKSM3}, an integral spectral representation of the propagator of the massive Dirac equation in the non-extreme Kerr geometry outside the event horizon is derived in Boyer--Lindquist coordinates. It has been used to study the long-time behavior (including decay rates) and the escape probability of Dirac particles in rotating Kerr black hole spacetimes \cite{FKSM2}. The shortcoming of this integral spectral representation is, however, that it yields a solution of the associated Cauchy problem only outside the event horizon. In the present paper, we construct a generalized integral spectral representation that describes the complete dynamics of Dirac particles outside, across, and inside the event horizon, up to the Cauchy horizon. The methods used in the derivation of our integral spectral representation are quite different from those employed in \cite{FKSM3}, as is now outlined. We work with horizon-penetrating advanced Eddington--Finkelstein-type coordinates \cite{R}, i.e., with an analytic extension of Boyer--Lindquist coordinates that simultaneously covers both the exterior and the interior black hole region without exhibiting singularities at the horizons and features a proper time function $\tau$. Furthermore, we employ a regular Carter tetrad, i.e., a symmetric Newman--Penrose null frame, which reflects, on the one hand, the discrete time and angle reversal isometries of the Kerr geometry and, on the other hand, its Petrov type. After computing the corresponding spin coefficients by solving the first Maurer--Cartan equation of structure, we explicitly determine the massive Dirac equation in Hamiltonian form 
\begin{equation*} 
\textnormal{i} \partial_{\tau} \psi(\tau, r, \theta, \phi) = H \psi(\tau, r, \theta, \phi) \, ,
\end{equation*}
where $H$ denotes the Hamiltonian, $\psi$ is a Dirac $4$-spinor, and $(\tau, r, \theta, \phi)$ are the advanced Eddington--Finkelstein-type coordinates. Moreover, we introduce a suitable scalar product on the associated space of solutions, and show for smooth and compactly supported Dirac $4$-spinors that the Hamiltonian is symmetric with respect to this scalar product. We also establish that it coincides with the canonical scalar product obtained by integrating the normal component (defined with respect to the level sets of $\tau$) of the Dirac current. As we apply the spectral theorem in the derivation of the propagator, we need to establish the essential self-adjointness of the Dirac Hamiltonian. To this end, we first impose a Dirichlet-type boundary condition on a time-like inner boundary surface placed beyond the Cauchy horizon. This boundary condition prevents Dirac particles from impinging on the curvature singularity without affecting their dynamics in the region outside the Cauchy horizon, consequently leading to a unitary time evolution. Then, we apply the method of proof for the essential self-adjointness of the Dirac Hamiltonian for mixed initial-boundary value problems that are not uniformly elliptic introduced in \cite{FRö}. Subsequently, it is possible to derive an integral representation of the Dirac propagator via the spectral theorem for unbounded self-adjoint operators 
\begin{equation*}
\psi = e^{- \textnormal{i} \tau H}\, \psi_0 = \int_{\mathbb{R}} e^{- \textnormal{i} \omega \tau} \, \psi_0 \, \textnormal{d}P_{\omega} \, , 
\end{equation*}
where $\textnormal{d}P_{\omega}$ is the spectral measure of $H$, $\omega$ is the spectral parameter, and $\psi_0 := \psi(\tau = 0, r, \theta, \phi)$ is smooth initial data with compact support. We compute the spectral measure of the Dirac Hamiltonian employing Stone's formula \cite{ReedSimon}, which yields an explicit expression in terms of the resolvent of the Hamiltonian $\textnormal{Res}(H) := (H - \omega_{\textnormal{c}})^{-1}$, for which $\omega_{\textnormal{c}} \in \mathbb{C}\backslash \mathbb{R}$. Furthermore, we determine the resolvent applying quantities obtained in the analysis of the systems of radial and angular ordinary differential equations (ODEs) that arise in Chandrasekhar's separation of variables, that is, after factoring out the azimuthal angle modes, we project the Dirac Hamiltonian onto a finite-dimensional, invariant spectral eigenspace of the angular operator, which leaves us with a matrix-valued first-order ordinary differential operator in the radial variable. The resolvent of this operator can be calculated by means of the Green's matrix of the radial ODE system. For this purpose, we derive generalized Jost-type equations \cite{ReedSimon3} for this system and study specific aspects of their solutions, namely the existence, uniqueness, and boundedness. We moreover use the asymptotic radial solutions at infinity, the event horizon, and the Cauchy horizon for guidance in the implicit construction of the fundamental solutions of the radial ODE system required for the computation of the corresponding Green's matrix. Eventually, by summing over all azimuthal angle modes, we obtain the full resolvent of the Dirac Hamiltonian in separated form. The resulting horizon-penetrating generalization of the integral spectral representation of the Dirac propagator describes the complete dynamics of massive Dirac particles outside and across the event horizon of the non-extreme Kerr geometry, up to the Cauchy horizon.\\[0.5ex]

\noindent \textbf{Main Theorem.} \textit{The massive Dirac propagator in the non-extreme Kerr geometry in horizon-penetrating advanced Eddington--Finkelstein-type coordinates can be expressed via the integral spectral representation 
\begin{equation*}
\psi(\tau, r, \theta, \phi) = \frac{1}{2 \pi \textnormal{i}} \, \sum_{k \in \mathbb{Z}} \, e^{- \textnormal{i} k \phi} \int_{\mathbb{R}} e^{- \textnormal{i} \omega \tau} \, \lim_{\epsilon \searrow 0} \, \bigl[(H_k - \omega - \textnormal{i} \epsilon)^{- 1} - (H_k - \omega + \textnormal{i} \epsilon)^{- 1}\bigr](r, \theta; r', \theta') \, \psi_{0, k}(r', \theta') \, \textnormal{d}\omega \, ,
\end{equation*} 
where $\psi_{0, k}$ is the initial data for fixed $k$-modes and $(H_k - \omega \mp \textnormal{i} \epsilon)^{- 1}$ are the unique resolvents of the Dirac Hamiltonian $H_k$ for fixed $k$-modes on the upper and lower complex half-planes.} \\[0.5ex]

\noindent In a final step, we compute the limit $\epsilon \searrow 0$ of the difference of resolvents leading to a rigorously simplified form of the propagator.

The article is organized as follows. In Section \ref{prel}, we provide the mathematical framework for the Kerr geometry and for the massive Dirac equation. Moreover, we recall required results from the asymptotic analysis of the radial ODE system and from the spectral analysis of the angular ODE system arising in Chandrasekhar's separation of variables without giving proofs. We derive the Hamiltonian formulation and a suitable scalar product for the space of solutions of the associated Cauchy problem in Sections \ref{III} and \ref{IV}, respectively. Furthermore, we verify the symmetry of the Hamiltonian with respect to this scalar product in Appendix \ref{appA}. In Section \ref{V}, we show the essential self-adjointness of the Hamiltonian. Finally, we construct the resolvent of the Hamiltonian and the integral spectral representation of the propagator in Section \ref{VI}. The fundamental radial solutions required for the computation of the resolvent are determined in Appendix \ref{appB}.

%-----------------------------------------------------
\section{Preliminaries} \label{prel}
%-----------------------------------------------------

\noindent We recall the necessary basics on the non-extreme Kerr geometry in horizon-penetrating advanced Eddington--Finkelstein-type coordinates, the general relativistic, massive Dirac equation in the Newman--Penrose formalism, and Chandrasekhar's separation of variables.

The non-extreme Kerr geometry is a connected, orientable and time-orientable, smooth, asymptotically flat Lorentzian 4-manifold $(\mathfrak{M}, \boldsymbol{g})$ with topology $S^2 \times \mathbb{R}^2$, for which the metric $\boldsymbol{g}$ is stationary and axisymmetric and given in horizon-penetrating advanced Eddington--Finkelstein-type coordinates $(\tau, r, \theta, \phi)$ with $\tau \in \mathbb{R}, r \in \mathbb{R}_{> 0}, \theta \in [0, \pi]$, and $\phi \in [0, 2 \pi)$ \cite{R} by 
\begin{equation} \label{Kerrmetric}
\begin{split}
\boldsymbol{g} & = \biggl(1 - \frac{2 M r}{\Sigma} \biggr) \textnormal{d}\tau \otimes \textnormal{d}\tau - \frac{2 M r}{\Sigma} \Bigl(\bigl[\textnormal{d}r - a \sin^2{(\theta)} \, \textnormal{d}\phi\bigr] \otimes \textnormal{d}\tau 
+ \textnormal{d}\tau \otimes \bigl[\textnormal{d}r - a \sin^2{(\theta)} \, \textnormal{d}\phi\bigr]\Bigr) \\ \\
& \hspace{0.4cm} - \biggl(1 + \frac{2 M r}{\Sigma}\biggr) \, \bigl(\textnormal{d}r - a \sin^2{(\theta)} \, \textnormal{d}\phi\bigr) \otimes \bigl(\textnormal{d}r - a \sin^2{(\theta)} \, \textnormal{d}\phi\bigr) - \Sigma \, \textnormal{d}\theta \otimes \textnormal{d}\theta - \Sigma \sin^2{(\theta)} \, \textnormal{d}\phi \otimes  \textnormal{d}\phi \, ,
\end{split}
\end{equation}
where $M$ is the mass and $a M$ the angular momentum of the black hole, with $0 \leq a < M$, and $\Sigma = \Sigma(r, \theta) := r^2 + a^2 \cos^2{(\theta)}$. The event and Cauchy horizons are located at $r_{\pm} := M \pm \sqrt{M^2 - a^2}$, respectively. The advanced Eddington--Finkelstein-type coordinates are an analytic extension of the common Boyer--Lindquist coordinates $(t, r, \theta, \varphi)$ with $t \in \mathbb{R}, r \in \mathbb{R}_{> 0}, \theta \in [0, \pi]$, and $\varphi \in [0, 2 \pi)$ \cite{BL}, covering both the exterior and interior black hole regions while being regular at the horizons. In terms of the Boyer--Lindquist coordinates, the advanced Eddington--Finkelstein-type time and azimuthal angle coordinates read
\begin{equation}\label{AEFSTC}
\begin{split}
\tau & := t + \frac{r_+^2 + a^2}{r_+ - r_-} \, \ln{|r - r_+|} - \frac{r_-^2 + a^2}{r_+ - r_-} \, \ln{|r - r_-|} \\ \\
\phi & := \varphi + \frac{a}{r_+ - r_-} \, \ln{\biggl|\frac{r - r_+}{r - r_-}\biggr|} \, . 
\end{split}
\end{equation}
This horizon-penetrating coordinate system possesses a proper time function, unlike the original advanced Eddington--Finkelstein (null) coordinates \cite{Edd, Fink}. It is advantageous to describe the Kerr geometry in the Newman--Penrose formalism using a regular Carter tetrad \cite{Carter, R}
\begin{equation}\label{CT}
\begin{split}
\boldsymbol{l} & = \frac{1}{\sqrt{2 \Sigma} \, r_+} \, \bigl([\Delta + 4 M r] \, \partial_{\tau} + \Delta \, \partial_r + 2 a \, \partial_{\phi}\bigr) \\ \\
\boldsymbol{n} & = \frac{r_+}{\sqrt{2 \Sigma}} \, (\partial_{\tau} - \partial_r) \\ \\
\boldsymbol{m} & = \frac{1}{\sqrt{2 \Sigma}} \bigl(\textnormal{i} a \sin{(\theta)} \, \partial_{\tau} + \partial_{\theta} + \textnormal{i} \csc{(\theta)} \, \partial_{\phi}\bigr) \\ \\
\boldsymbol{\overline{m}} & = - \frac{1}{\sqrt{2 \Sigma}} \bigl(\textnormal{i} a \sin{(\theta)} \, \partial_{\tau} - \partial_{\theta} + \textnormal{i} \csc{(\theta)} \, \partial_{\phi}\bigr) \, , 
\end{split}
\end{equation}
with $\Delta = \Delta(r) := (r - r_+) (r - r_-) = r^2 - 2 M r + a^2$ being the horizon function, because this frame is adapted to the two principal null directions of the Weyl tensor and to the fundamental discrete time and angle reversal isometries. Thus, since the Kerr geometry is algebraically special and of Petrov type D, one has the computational advantage that the four spin coefficients $\kappa, \sigma, \lambda$, and $\nu$ as well as the four Weyl scalars $\Psi_0, \Psi_1, \Psi_3$, and $\Psi_4$ vanish \cite{BON}, and that specific spin coefficients are linearly dependent. Substituting the Carter tetrad (\ref{CT}) into -- and solving -- the first Maurer--Cartan equation of structure, we obtain the spin coefficients \cite{R}
\begin{equation}\label{SC}
\begin{split}
\kappa & = \sigma = \lambda = \nu = 0 \, , \,\,\,\,\,\,\,\, \gamma = - \frac{r_+}{2^{3/2} \sqrt{\Sigma} \, \bigl(r - \textnormal{i} a \cos{(\theta)}\bigr)} \, , \,\,\,\,\,\,\,\, \epsilon = \frac{r^2 - a^2 - 2 \textnormal{i} a \cos{(\theta)} \, (r - M)}{2^{3/2} \sqrt{\Sigma} \, r_+ \bigl(r - \textnormal{i} a \cos{(\theta)}\bigr)}  \\ \\
\pi & = - \tau = \frac{\textnormal{i} a \sin{(\theta)}}{\sqrt{2 \Sigma} \, \bigl(r - \textnormal{i} a \cos{(\theta)}\bigr)} \, , \,\,\,\,\,\,\,\, \mu = - \frac{r_+}{\sqrt{2 \Sigma} \, \bigl(r - \textnormal{i} a \cos{(\theta)}\bigr)} \, , \,\,\,\,\,\,\,\, \varrho = - \frac{\Delta}{\sqrt{2 \Sigma} \, r_+ \bigl(r - \textnormal{i} a \cos{(\theta)}\bigr)} \\ \\
\alpha & = - \beta = - \frac{1}{(2 \Sigma)^{3/2}} \, \bigl[\bigl(r^2 + a^2\bigr) \cot{(\theta)} - \textnormal{i} r a \sin{(\theta)}\bigr] \, .
\end{split}
\end{equation}

Next, introducing a spin bundle $S\mathfrak{M} = \mathfrak{M} \times \mathbb{C}^4$ on $\mathfrak{M}$ with fibers $S_{\boldsymbol{x}}\mathfrak{M} \simeq \mathbb{C}^4$, $\boldsymbol{x} \in \mathfrak{M}$, we can formulate the general relativistic, massive Dirac equation (without an external potential) 
\begin{equation} \label{mde} 
\bigl(\gamma^{\mu} \nabla_{\mu} + \textnormal{i} m\bigr) \psi(x^{\mu}) = \boldsymbol{0} \, , \,\,\,\,\, \mu \in \{0, 1, 2, 3\} \, ,
\end{equation}
where $\boldsymbol{\nabla}$ is the metric connection on $S\mathfrak{M}$, $\gamma^{\mu}$ are the Dirac matrices, $\psi$ is the Dirac $4$-spinor defined on the fibers $S_{\boldsymbol{x}}\mathfrak{M}$, and $m$ is the invariant fermion mass. In the Newman--Penrose formalism -- by employing a local dyad spinor frame -- (\ref{mde}) becomes the coupled first-order system of partial differential equations 
\begin{equation}\label{DE}
\begin{split}
(n^{\mu} \partial_{\mu} + \overline{\mu} - \overline{\gamma}) \, \mathscr{G}_1 - (\overline{m}^{\, \mu} \partial_{\mu} + \overline{\beta} - \overline{\tau}) \, \mathscr{G}_2 & = \frac{\textnormal{i} m}{\sqrt{2}} \, \mathscr{F}_1 \\
(l^{\mu} \partial_{\mu} + \overline{\varepsilon} - \overline{\varrho}) \, \mathscr{G}_2 - (m^{\mu} \partial_{\mu} + \overline{\pi} - \overline{\alpha}) \, \mathscr{G}_1 & = \frac{\textnormal{i} m}{\sqrt{2}} \, \mathscr{F}_2 \\ 
(l^{\mu} \partial_{\mu} + \varepsilon - \varrho) \, \mathscr{F}_1 + (\overline{m}^{\, \mu} \partial_{\mu} + \pi - \alpha) \, \mathscr{F}_2 & = \frac{\textnormal{i} m}{\sqrt{2}} \, \mathscr{G}_1 \\ 
(n^{\mu} \partial_{\mu} + \mu - \gamma) \, \mathscr{F}_2 + (m^{\mu} \partial_{\mu} + \beta - \tau) \, \mathscr{F}_1 & = \frac{\textnormal{i} m}{\sqrt{2}} \, \mathscr{G}_2 
\end{split}
\end{equation}
with $\psi = (\mathscr{F}_1, \mathscr{F}_2, - \mathscr{G}_1, - \mathscr{G}_2)^{\textnormal{T}}$ \cite{ChandraBook}. Inserting the Carter tetrad (\ref{CT}) and the associated spin coefficients (\ref{SC}) into the system (\ref{DE}), and applying the transformation
\begin{equation}\label{SpinTrafo}
\psi' = \mathscr{P} \psi = (\mathscr{H}_1, \mathscr{H}_2, - \mathscr{J}_1, - \mathscr{J}_2)^{\textnormal{T}} \, , \quad \gamma'^{\mu} = \mathscr{P} \gamma^{\mu} \mathscr{P}^{- 1} \, ,
\end{equation}
where
\begin{equation} \label{spinortr}
\mathscr{P} := \textnormal{diag}\Bigl(\sqrt{r - \textnormal{i} a \cos{(\theta)}}, \sqrt{r - \textnormal{i} a \cos{(\theta)}}, \sqrt{r + \textnormal{i} a \cos{(\theta)}}, \sqrt{r + \textnormal{i} a \cos{(\theta)}} \, \Bigr) \, ,
\end{equation}
we find
\begin{equation}\label{KDE}
\begin{split}
& r_+ \bigl(\partial_{\tau} - \partial_r\bigr) \mathscr{J}_1 + \bigl(\textnormal{i} a \sin{(\theta)} \, \partial_{\tau} - \partial_{\theta} + \textnormal{i} \csc{(\theta)} \, \partial_{\phi} - 2^{- 1} \cot{(\theta)}\bigr) \mathscr{J}_2 = \textnormal{i} m \bigl(r + \textnormal{i} a \cos{(\theta)}\bigr) \mathscr{H}_1 \\ \\
& r_+^{- 1} \bigl([\Delta + 4 M r] \, \partial_{\tau} + \Delta \, \partial_r + 2 a \, \partial_{\phi} + r - M\bigr) \mathscr{J}_2 - \bigl(\textnormal{i} a \sin{(\theta)} \, \partial_{\tau} + \partial_{\theta} + \textnormal{i} \csc{(\theta)} \, \partial_{\phi} + 2^{- 1} \cot{(\theta)}\bigr) \mathscr{J}_1 \\
& = \textnormal{i} m \bigl(r + \textnormal{i} a \cos{(\theta)}\bigr) \mathscr{H}_2 \\ \\
& r_+^{- 1} \bigl([\Delta + 4 M r] \, \partial_{\tau} + \Delta \, \partial_r + 2 a \, \partial_{\phi} + r - M\bigr) \mathscr{H}_1 - \bigl(\textnormal{i} a \sin{(\theta)} \, \partial_{\tau} - \partial_{\theta} + \textnormal{i} \csc{(\theta)} \, \partial_{\phi} - 2^{- 1} \cot{(\theta)}\bigr) \mathscr{H}_2 \\
& = \textnormal{i} m \bigl(r - \textnormal{i} a \cos{(\theta)}\bigr) \mathscr{J}_1 \\ \\
& r_+ \bigl(\partial_{\tau} - \partial_r\bigr) \mathscr{H}_2 + \bigl(\textnormal{i} a \sin{(\theta)} \, \partial_{\tau} + \partial_{\theta} + \textnormal{i} \csc{(\theta)} \, \partial_{\phi} + 2^{- 1} \cot{(\theta)}\bigr) \mathscr{H}_1 = \textnormal{i} m \bigl(r - \textnormal{i} a \cos{(\theta)}\bigr) \mathscr{J}_2 \, ,
\end{split}
\end{equation}
which is the starting point for the derivation of the Hamiltonian formulation of the massive Dirac equation on a Kerr background geometry in horizon-penetrating coordinates presented in the next section. We note in passing that the system (\ref{KDE}) corresponds to the transformed Dirac equation 
\begin{equation}\label{NSFDE}
- \sqrt{\Sigma} \, \gamma^0 \mathscr{P}^{\dagger} \mathscr{P}^{- 1} \Bigl(\gamma'^{\mu} \bigl[\nabla_{\mu} + \mathscr{P} \, \partial_{\mu} \bigl(\mathscr{P}^{- 1}\bigr)\bigr] + \textnormal{i} m\Bigr) \psi' = \boldsymbol{0} \, ,
\end{equation} 
where $\gamma^0 := \textnormal{diag}(1, 1, - 1, - 1)$. This will become relevant in the following construction of both the Hamiltonian formulation and the scalar product.

Finally, for the explicit computation of the resolvent of the Dirac Hamiltonian, we require specific results arising from Chandrasekhar's separation of variables of the system (\ref{KDE}). More precisely, we need the asymptotic solutions of the radial ODE system at infinity, at the event horizon, and at the Cauchy horizon, as well as certain information about the eigenvalues and eigenfunctions of the angular ODE system. In the following, these results are recalled. For a detailed analysis and proofs see \cite{R}. Substituting the separation ansatz
\begin{equation} \label{Chandrasepans}
\begin{split}
\mathscr{H}_1 & = e^{- \textnormal{i} (\omega \tau + k \phi)} \, \mathscr{R}_+(r) \mathscr{T}_+(\theta) \\
\mathscr{H}_2 & = e^{- \textnormal{i} (\omega \tau + k \phi)} \, \mathscr{R}_-(r) \mathscr{T}_-(\theta) \\
\mathscr{J}_1 & = e^{- \textnormal{i} (\omega \tau + k \phi)} \, \mathscr{R}_-(r) \mathscr{T}_+(\theta) \\
\mathscr{J}_2 & = e^{- \textnormal{i} (\omega \tau + k \phi)} \, \mathscr{R}_+(r) \mathscr{T}_-(\theta) \, ,
\end{split}
\end{equation}
in which $\omega \in \mathbb{R}$ and $k \in \mathbb{Z} + 1/2$, into (\ref{KDE}) yields the first-order radial and angular ODE systems 
\begin{align}
R(r_{\star}) \, \widetilde{\mathscr{R}} & = \displaystyle \frac{\sqrt{|\Delta|}}{r^2 + a^2} \left(\begin{array}{cc}
0 & 1 \\
\textnormal{sign}(\Delta) & 0  \label{RADODESYS} \\
\end{array}\right) \xi \, \widetilde{\mathscr{R}} \\ 
A(\theta) \, \mathscr{T} & = \xi \, \mathscr{T} \, , \nonumber
\end{align} 
where
\begin{equation*}
r_{\star} := r + \frac{r_+^2 + a^2}{r_+ - r_-} \, \ln{|r - r_+|} - \frac{r_-^2 + a^2}{r_+ - r_-} \, \ln{|r - r_-|}   
\end{equation*}
is the Regge--Wheeler coordinate,
\begin{align}
R(r_{\star}) & := \1_{\mathbb{C}^2} \, \partial_{r_{\star}} + \frac{\textnormal{i}}{r^2 + a^2} \left(\begin{array}{cc}
- \omega (\Delta + 4 M r) - 2 a k & \hspace{0.3cm} - \sqrt{|\Delta|} \, m r \\ \\ 
\sqrt{|\Delta|} \,\, \textnormal{sign}(\Delta) \, m r & \omega \Delta \\
\end{array}\right) \nonumber \\ \nonumber \\
A(\theta) & := \left(\begin{array}{cc}
m a \cos{(\theta)} & - \partial_{\theta} - 2^{- 1} \cot{(\theta)} + a \omega \sin{(\theta)} + k \, \textnormal{csc}(\theta) \\ \\
\partial_{\theta} + 2^{- 1} \cot{(\theta)} + a \omega \sin{(\theta)} + k \, \textnormal{csc}(\theta) & - m a \cos{(\theta)} 
\end{array}\right)  \label{angop}
\end{align}
are matrix-valued radial and angular operators, $\widetilde{\mathscr{R}} := \bigl(\sqrt{|\Delta|} \, \mathscr{R}_+, r_+ \, \mathscr{R}_-\bigr)^{\textnormal{T}}$ and $\mathscr{T} := (\mathscr{T}_+, \mathscr{T}_-)^{\textnormal{T}}$ are radial and angular vector-valued functions, and $\xi$ is the constant of separation. The asymptotic solutions and the decay properties of the associated errors of the radial ODE system at infinity, the event horizon, and the Cauchy horizon are specified in the lemmas below.
\begin{Lemma} \label{L1}
Every nontrivial solution $\widetilde{\mathscr{R}}$ of (\ref{RADODESYS}) for $|\omega| \geq m$ is asymptotically as $r \rightarrow \infty$ of the oscillatory form
\begin{equation*} 
\widetilde{\mathscr{R}}(r_{\star}) = \widetilde{\mathscr{R}}_{\infty}(r_{\star}) + E_{\infty}(r_{\star}) = D_{\infty}  \left(\begin{array}{c}
\mathfrak{f}_{\infty}^{(1)} \, e^{\textnormal{i} \phi_+(r_{\star})} \\ 
\mathfrak{f}_{\infty}^{(2)} \, e^{- \textnormal{i} \phi_-(r_{\star})}
\end{array}\right) + E_{\infty}(r_{\star}) \, ,
\end{equation*}	
where 	
\begin{equation*} 
D_{\infty} = 
\left(\begin{array}{cc}
\cosh{(\Omega)} & \sinh{(\Omega)} \\ 
\sinh{(\Omega)} & \cosh{(\Omega)}
\end{array}\right) 
\quad \textnormal{with} \quad
\Omega := \frac{1}{4} \ln{\biggl(\frac{\omega - m}{\omega + m}\biggr)} \, ,  
\end{equation*}
the functions 
\begin{equation} \label{asypha}
\phi_{\pm}(r_{\star}) = \textnormal{sign}(\omega) \biggl[- \sqrt{\omega^2 - m^2} \, r_{\star} + M \biggl(\pm \, 2 \omega - \frac{m^2}{\sqrt{\omega^2 - m^2}}\biggr) \ln{(r_{\star})}\biggr] 
\end{equation}
are the asymptotic phases, and $\boldsymbol{\mathfrak{f}}_{\infty} = \bigl(\mathfrak{f}_{\infty}^{(1)}, \mathfrak{f}_{\infty}^{(2)}\bigr)^{\textnormal{T}} \not= \boldsymbol{0}$ is a vector-valued constant. The error $E_{\infty}$ has polynomial decay
\begin{equation*} 
\|E_{\infty}(r_{\star})\| = \bigl\|\widetilde{\mathscr{R}}(r_{\star}) - \widetilde{\mathscr{R}}_{\infty}(r_{\star})\bigr\| \leq \frac{a}{r_{\star}} 
\end{equation*}
for a suitable constant $a \in \mathbb{R}_{> 0}$. For the case $|\omega| < m$, the non-trivial solution $\widetilde{\mathscr{R}}$ has both contributions that show exponential decay $\sim e^{- \sqrt{m^2 - \omega^2} \, r_{\star}}$ and exponential growth $\sim e^{\sqrt{m^2 - \omega^2} \, r_{\star}}$.
\end{Lemma}
\begin{Lemma} \label{L2a3}
Every nontrivial solution $\widetilde{\mathscr{R}}$ of (\ref{RADODESYS}) is asymptotically as $r \searrow r_{\pm}$ of the form
\begin{equation*} 
\widetilde{\mathscr{R}}(r_{\star}) = \widetilde{\mathscr{R}}_{r_{\pm}}(r_{\star}) + E_{r_{\pm}}(r_{\star}) = \left(\begin{array}{c}
\mathfrak{g}_{r_{\pm}}^{(1)} \, e^{\bigl(2 \textnormal{i} \bigl[\omega + k \Omega^{(\pm)}_{\textnormal{Kerr}}\bigr] r_{\star}\bigr)} \\ 
\mathfrak{g}_{r_{\pm}}^{(2)}
\end{array}\right) + E_{r_{\pm}}(r_{\star})
\end{equation*}
with the constants $\boldsymbol{\mathfrak{g}}_{r_{\pm}} = \bigl(\mathfrak{g}_{r_{\pm}}^{(1)}, \mathfrak{g}_{r_{\pm}}^{(2)}\bigr)^{\textnormal{T}} \not= \boldsymbol{0}$ and $\Omega^{(\pm)}_{\textnormal{Kerr}} := a/(2 M r_{\pm})$, as well as an error with exponential decay 
\begin{equation*}
\|E_{r_{\pm}}(r_{\star})\| = \bigl\|\widetilde{\mathscr{R}}(r_{\star}) - \widetilde{\mathscr{R}}_{r_{\pm}}(r_{\star})\bigr\| \leq p_{\pm} \, e^{\pm q_{\pm} r_\star}
\end{equation*}
for suitable constants $p_{\pm}, q_{\pm} \in \mathbb{R}_{> 0}$.
\end{Lemma}

\noindent The spectral properties of the eigenvalues and eigenfunctions of the angular ODE system are summarized in the following proposition.
\begin{Prp} For any $\omega \in \mathbb{R}$ and $k \in \mathbb{Z} + 1/2$, the differential operator (\ref{angop}) has a complete set of orthonormal eigenfunctions $(\mathscr{T}_l)_{l \in \mathbb{Z}}$ in $L^2 \bigl((0, \pi), \sin{(\theta)} \, \textnormal{d}\theta\bigr)^2$. The corresponding eigenvalues $\xi_l$ are real-valued and non-degenerate, and can thus be ordered as $\xi_l < \xi_{l + 1}$. Moreover, the eigenfunctions are pointwise bounded and smooth away from the poles,
\begin{equation*}
\mathscr{T}_l \in L^\infty\bigl((0, \pi)\bigr)^2 \cap C^\infty \bigl((0, \pi)\bigr)^2 \, .
\end{equation*}
Both the eigenfunctions $\mathscr{T}_l$ and the eigenvalues $\xi_l$ depend smoothly on $\omega$.
\end{Prp}

%-------------------------------------------------------------------------------------------------------------
\section{Hamiltonian Formulation of the Massive Dirac Equation in the Non-extreme Kerr Geometry in Horizon-penetrating Coordinates} \label{III}
%-------------------------------------------------------------------------------------------------------------

\noindent For the derivation of the Hamiltonian formulation of the massive Dirac equation in the non-extreme Kerr geometry in horizon-penetrating advanced Eddington--Finkelstein-type coordinates, it is advantageous to first rewrite the system (\ref{KDE}) in the form
\begin{equation*} 
(\mathcal{R} + \mathcal{A}) \psi' = \boldsymbol{0} \, ,
\end{equation*}
where
\begin{equation}\label{MOR}
\mathcal{R} := \left(\begin{array}{cccc}
- \textnormal{i} m r & 0 & - \mathscr{D}_-& 0 \\
0 & - \textnormal{i} m r & 0 & - \mathscr{D}_+ \\
\mathscr{D}_+ & 0 & \textnormal{i} m r & 0 \\
0 & \mathscr{D}_- & 0 & \textnormal{i} m r 
\end{array}\right) 
\end{equation}
and
\begin{equation}\label{MOA}
\mathcal{A} := \left(\begin{array}{cccc}
m a \cos{(\theta)} & 0 & 0 & \overline{\mathscr{L}} \\
0 & m a \cos{(\theta)} & \mathscr{L} & 0 \\
0 & \overline{\mathscr{L}} & m a \cos{(\theta)} & 0 \\
\mathscr{L} & 0 & 0 & m a \cos{(\theta)} 
\end{array}\right)
\end{equation}
are matrix-valued differential operators with
\begin{equation*}
\begin{split}
\mathscr{D}_+ & := r_+^{- 1} \, \bigl([\Delta + 4 M r] \, \partial_{\tau} + \Delta \, \partial_{r} + 2 a \, \partial_{\phi} + r - M\bigr) \\ \\
\mathscr{D}_- & := r_+ (\partial_{\tau} - \partial_{r}) \\ \\
\mathscr{L} & := \textnormal{i} a \sin{(\theta)} \, \partial_{\tau} + \partial_{\theta} + \textnormal{i} \, \textnormal{csc}(\theta) \, \partial_{\phi} + 2^{- 1} \cot{(\theta)} \, .
\end{split}
\end{equation*}
We then separate the $\tau$-derivative and multiply the resulting equation by the inverse of the matrix
\begin{equation} \label{trdm}
\widetilde{\gamma}'^{\tau} := - \sqrt{\Sigma} \, \gamma^0 \mathscr{P}^{\dagger} \mathscr{P}^{- 1} \gamma'^{\tau} = \left(\begin{array}{cccc}
0 & 0 & - r_+ & - \textnormal{i} a \sin{(\theta)} \\
0 & 0 & \textnormal{i} a \sin{(\theta)} & - r_+^{- 1} [\Delta + 4 M r] \\
r_+^{- 1} [\Delta + 4 M r] & - \textnormal{i} a \sin{(\theta)} & 0 & 0 \\
\textnormal{i} a \sin{(\theta)} & r_+ & 0 & 0
\end{array}\right)
\end{equation}
(cf.\ Eq.\ (\ref{NSFDE})) as well as by the imaginary unit, which leads to the Schr\"odinger-type equation
\begin{equation}\label{DEHF}
\textnormal{i} \partial_{\tau} \psi' = - \textnormal{i} \, (\widetilde{\gamma}'^{\tau})^{- 1} \bigl(\mathcal{R}^{(3)} + \mathcal{A}^{(3)}\bigr) \psi' =: H \psi',
\end{equation}
where $\mathcal{R}^{(3)}$ and $\mathcal{A}^{(3)}$ contain the first-order spatial and all zero-order contributions of the operators (\ref{MOR}) and (\ref{MOA}), respectively. The Dirac Hamiltonian $H$ may be recast in the more convenient form 
\begin{equation} \label{DHCF}
H = \alpha^j \partial_j + \mathscr{V} \, , \quad j \in \{r, \theta, \phi\} \, ,
\end{equation}
with the matrix-valued coefficients 
\begin{align}
\alpha^r & := - \frac{1}{\Sigma + 2 M r} \, \left(\begin{array}{cccc}
\textnormal{i} \Delta & r_+ a \sin{(\theta)} & 0 & 0 \\
r_+^{- 1} \Delta \, a \sin{(\theta)} & - \textnormal{i} \, (\Delta + 4 M r) & 0 & 0 \\
0 & 0 & - \textnormal{i} \, (\Delta + 4 M r) & r_+^{- 1} \Delta \, a \sin{(\theta)}\\
0 & 0 & r_+ a \sin{(\theta)} & \textnormal{i} \Delta 
\end{array}\right) \label{Br}\\ \nonumber \\
\alpha^{\theta} & := - \frac{1}{\Sigma + 2 M r} \, \left(\begin{array}{cccc}
- a \sin{(\theta)} & \textnormal{i} r_+ & 0 & 0 \\
\textnormal{i} r_+^{- 1} [\Delta + 4 M r]& a \sin{(\theta)} & 0 & 0 \\
0 & 0 & - a \sin{(\theta)} & - \textnormal{i} r_+^{- 1} [\Delta + 4 M r] \\
0 & 0 & - \textnormal{i} r_+  &  a \sin{(\theta)}
\end{array}\right) \\ \nonumber \\
\alpha^{\phi} & := - \frac{1}{\Sigma + 2 M r} \, \left(\begin{array}{cccc}
\textnormal{i} a & r_+ \textnormal{csc}(\theta) & 0 & 0 \\
r_+^{- 1} \textnormal{csc}(\theta) \, (\Delta - 2 \Sigma) & - \textnormal{i} a & 0 & 0 \\
0 & 0 & - \textnormal{i} a & r_+^{- 1} \textnormal{csc}(\theta) \, (\Delta - 2 \Sigma) \\
0 & 0 & r_+ \textnormal{csc}(\theta) & \textnormal{i} a 
\end{array}\right) \label{Bphi}
\end{align}
and the potential
\begin{equation}
\begin{split}\label{PotV}
\mathscr{V} := - \frac{1}{\Sigma + 2 M r} \,
&\left(\begin{array}{cc}
\mathscr{B}_1 & \mathscr{B}_2 \\
\mathscr{B}_3 & \mathscr{B}_4
\end{array}\right),
\end{split}
\end{equation}
where the quantities $\mathscr{B}_k$, $k \in \{1, 2, 3, 4\}$, are the $(2 \times 2)$-blocks
\begin{equation} \label{B1234}
\begin{split}
\mathscr{B}_1 & := \frac{1}{2} \left(\begin{array}{cc}
2 \textnormal{i} \, (r - M) - a \cos{(\theta)} & \textnormal{i} r_+ \cot{(\theta)} \\ \\
r_+^{- 1} \bigl[2 a \sin{(\theta)} \, (r - M) + \textnormal{i} \, \cot{(\theta)} \, (\Delta + 4 M r) \bigr] & a \cos{(\theta)} 
\end{array}\right) \\ \\
\mathscr{B}_2 & := - m \, \bigl(r - \textnormal{i} a \cos{(\theta)}\bigr) \left(\begin{array}{cc} 
r_+ & \textnormal{i} a \sin{(\theta)} \\ \\
- \textnormal{i} a \sin{(\theta)} & r_+^{- 1} (\Delta + 4 M r)  
\end{array}\right) \\ \\
\mathscr{B}_3 & := - m \, \bigl(r + \textnormal{i} a \cos{(\theta)}\bigr) \left(\begin{array}{cc}
r_+^{- 1} (\Delta + 4 M r) & - \textnormal{i} a \sin{(\theta)} \\ \\
\textnormal{i} a \sin{(\theta)} & r_+ 
\end{array}\right) \\ \\
\mathscr{B}_4 & := \frac{1}{2} \left(\begin{array}{cc}
- a \cos{(\theta)} & r_+^{- 1} \bigl[2 a \, \sin{(\theta)} \, (r - M) - \textnormal{i} \, \cot{(\theta)} \, (\Delta + 4 M r)\bigr] \\ \\
- \textnormal{i} r_+ \cot{(\theta)} & 2 \textnormal{i} \, (r - M) + a \cos{(\theta)}
\end{array}\right).
\end{split}
\end{equation}

%-------------------------------------------------------------------------------------------------------------
\section{The Canonical Scalar Product} \label{IV}
%-------------------------------------------------------------------------------------------------------------

\noindent In order to set up a Hilbert space that contains the solutions of (\ref{DEHF})
\begin{equation*}
\mathcal{H} := \overline{\bigl(\textnormal{Sol}(H - \textnormal{i} \partial_{\tau}), \, (\, \cdot \, | \, \cdot \,)\bigr)} \quad \textnormal{for which} \quad \textnormal{Sol}(H - \textnormal{i} \partial_{\tau}) = \bigl\{\psi' \in L^2(\mathfrak{M}, S\mathfrak{M}) \, | \, (H - \textnormal{i} \partial_{\tau}) \, \psi' = \boldsymbol{0}\bigr\} \, ,
\end{equation*}
and to establish the symmetry property of the Hamiltonian $H$
\begin{equation*}
(\phi' | H \psi') = (H \phi' | \psi') \quad \textnormal{with} \quad \phi', \psi' \in \textnormal{Sol}(H - \textnormal{i} \partial_{\tau})
\end{equation*}
(or its self-adjointness), we require a suitable scalar product $(\, \cdot \, | \, \cdot \,)$. We thus work with the scalar product \cite{FKSM3}
\begin{equation} \label{GSP}
(\psi | \phi)_{|\mathfrak{N}_{\tau}} := \int_{\mathfrak{N}_{\tau}} \Sl \psi | \boldsymbol{\slashed{\nu}} \phi \Sr \, \textnormal{d}\mu_{|\mathfrak{N}_{\tau}} 
\end{equation}
defined on the space-like hypersurface $\mathfrak{N}_{\tau} := \{\tau = \textnormal{const.}, r, \theta, \phi\}$, where
\begin{equation}\label{spinsp}
\Sl \cdot | \cdot \Sr : \,\, S_{\boldsymbol{x}}\mathfrak{M} \times S_{\boldsymbol{x}}\mathfrak{M} \rightarrow \mathbb{C} \, , \quad (\psi, \phi) \mapsto \psi^{\star} \phi
\end{equation}
denotes the indefinite spin scalar product of signature $(2, 2)$, $\psi^{\star} := \psi^{\dagger} \mathscr{S}$ the adjoint Dirac spinor, $\boldsymbol{\slashed{\nu}} = \gamma^{\mu} \nu_{\mu}$ the Clifford contraction of the future-directed, time-like normal $\boldsymbol{\nu}$, and $\textnormal{d}\mu_{|\mathfrak{N}_{\tau}} = \sqrt{|\textnormal{det}(\boldsymbol{g}_{|\mathfrak{N}_{\tau}})|} \, \textnormal{d}\phi \, \textnormal{d}\theta \, \textnormal{d}r$ is the invariant measure on $\mathfrak{N}_{\tau}$, in which $\boldsymbol{g}_{|\mathfrak{N}_{\tau}}$ is the induced Riemannian metric. The matrix $\mathscr{S}$ is defined via the relation 
\begin{equation} \label{defeq}
\gamma^{\mu \dagger} := \mathscr{S} \gamma^{\mu} \mathscr{S}^{- 1} \, .
\end{equation}
We note that this scalar product is independent of the choice of the specific space-like hypersurface $\mathfrak{N}_{\tau}$. This can be easily shown by means of Gauss' theorem and current conservation. In the following, we explicitly compute the above quantities and subsequently derive a more convenient representation for the scalar product. We begin with the calculation of the matrix $\mathscr{S}$. Via (\ref{trdm}) and the spinor transformation (\ref{SpinTrafo}), we find 
\begin{equation*}
\gamma^{\mu} = - \frac{1}{\sqrt{\Sigma}} \, (\mathscr{P}^{\dagger})^{- 1} \, \gamma^0 \, \widetilde{\gamma}'^{\mu} \, \mathscr{P} 
\end{equation*}
and hence, using the defining equation (\ref{defeq}), we obtain the expression
\begin{equation*}
\mathscr{S} = \left(\begin{array}{cccc}
0 & 0 & 1 & 0 \\
0 & 0 & 0 & 1 \\
1 & 0 & 0 & 0 \\
0 & 1 & 0 & 0
\end{array}\right).
\end{equation*}
Next, we determine the normal vector field $\boldsymbol{\nu}$ by means of the conditions 
\begin{equation*}
\left\langle \boldsymbol{\nu} | \partial_{r} \right\rangle_{|\boldsymbol{g}} = \left\langle \boldsymbol{\nu} | \partial_{\theta} \right\rangle_{|\boldsymbol{g}} = \left\langle \boldsymbol{\nu} | \partial_{\phi} \right\rangle_{|\boldsymbol{g}} = 0 \,\,\,\,\,\,\,\, \textnormal{and} \,\,\,\,\,\,\, \left\langle \boldsymbol{\nu} | \boldsymbol{\nu} \right\rangle_{|\boldsymbol{g}} = 1 \, ,
\end{equation*}
where $\left\langle \, \cdot \, | \, \cdot \, \right\rangle_{\boldsymbol{g}} := \boldsymbol{g}(\,\cdot\,,\,\cdot\,) = g_{\mu \nu} \, \textnormal{d}x^\mu \otimes \textnormal{d}x^\nu(\,\cdot\,,\,\cdot\,)$ is the spacetime inner product on $\mathfrak{M}$. Accordingly, employing (\ref{Kerrmetric}) yields
\begin{equation*}
\boldsymbol{\nu} = \biggl(1 + \frac{2 M r}{\Sigma}\biggr)^{1/2} \,\partial_{\tau} - \frac{2 M r}{\Sigma} \, \biggl(1 + \frac{2 M r}{\Sigma}\biggr)^{-1/2} \, \partial_r \, .
\end{equation*}
The corresponding dual co-vector reads 
\begin{equation}\label{DNV}
\boldsymbol{\nu} = \biggl(1 + \frac{2 M r}{\Sigma}\biggr)^{-1/2} \, \textnormal{d}\tau \, .
\end{equation}
Moreover, the induced metric $\boldsymbol{g}_{|\mathfrak{N}_{\tau}}$ on the hypersurface $\mathfrak{N}_{\tau}$ is simply the restriction of (\ref{Kerrmetric}) to $\mathfrak{N}_{\tau}$ and, thus, we obtain
\begin{equation*}
\boldsymbol{g}_{|\mathfrak{N}_{\tau}} = - \biggl(1 + \frac{2 M r}{\Sigma}\biggr) \bigl(\textnormal{d}r - a \sin^2{(\theta)} \, \textnormal{d}\phi\bigr) \otimes \bigl(\textnormal{d}r - a \sin^2{(\theta)} \, \textnormal{d}\phi\bigr) - \Sigma \, \textnormal{d}\theta \otimes \textnormal{d}\theta - \Sigma \sin^2{(\theta)} \, \textnormal{d}\phi \otimes \textnormal{d}\phi \, .
\end{equation*}
The associated Jacobian determinant in the volume measure $\textnormal{d}\mu_{|\mathfrak{N}_{\tau}}$ becomes
\begin{equation}\label{JD}
\sqrt{|\textnormal{det}(\boldsymbol{g}_{|\mathfrak{N}_{\tau}})|} = \Sigma \sin{(\theta)} \biggl(1 + \frac{2 M r}{\Sigma}\biggr)^{1/2} \, .
\end{equation}
We now express the scalar product (\ref{GSP}) in terms of the primed quantities (\ref{SpinTrafo}) used in (\ref{DEHF}), and substitute (\ref{DNV}) as well as (\ref{JD}). This results in 
\begin{equation} \label{GSP2}
(\psi' | \phi')_{|\mathfrak{N}_{\tau}} = \iiint \psi'^{\dagger} \mathscr{S}' \gamma'^{\tau} \phi' \, \Sigma \sin{(\theta)} \, \textnormal{d}\phi \, \textnormal{d}\theta \, \textnormal{d}r \, .
\end{equation}
Again employing (\ref{trdm}), that is with $\gamma'^{\tau} = - \mathscr{P} (\mathscr{P}^{\dagger})^{- 1} \gamma^0 \,  \widetilde{\gamma}'^{\tau}/\sqrt{\Sigma}$, the scalar product (\ref{GSP2}) yields
\begin{equation}\label{GSP3}
\begin{split}
(\psi' | \phi')_{|\mathfrak{N}_{\tau}} & = - \iiint \psi'^{\dagger} \, \mathscr{S}' \mathscr{P} (\mathscr{P}^{\dagger})^{- 1} \gamma^0 \,  \widetilde{\gamma}'^{\tau} \phi' \, \sqrt{\Sigma} \sin{(\theta)} \, \textnormal{d}\phi \, \textnormal{d}\theta \, \textnormal{d}r \\ \\
& = - \iiint \psi'^{\dagger} \, \mathscr{P} \mathscr{P}^{\dagger} \mathscr{S}' \mathscr{P} (\mathscr{P}^{\dagger})^{- 1} \gamma^0 \,  \widetilde{\gamma}'^{\tau} \phi' \, \sin{(\theta)} \, \textnormal{d}\phi \, \textnormal{d}\theta \, \textnormal{d}r \\ \\
& = - \iiint \psi'^{\dagger} \, \mathscr{P} \mathscr{S} (\mathscr{P}^{\dagger})^{- 1} \gamma^0 \,  \widetilde{\gamma}'^{\tau} \phi' \, \sin{(\theta)} \, \textnormal{d}\phi \, \textnormal{d}\theta \, \textnormal{d}r \\ \\
& = - \iiint \psi'^{\dagger} \, \mathscr{S} \mathscr{P}^{\dagger} (\mathscr{P}^{\dagger})^{- 1} \gamma^0 \,  \widetilde{\gamma}'^{\tau} \phi' \, \sin{(\theta)} \, \textnormal{d}\phi \, \textnormal{d}\theta \, \textnormal{d}r \\ \\
& = \iiint \psi'^{\dagger} \, \Gamma^{\tau} \phi' \, \sin{(\theta)} \, \textnormal{d}\phi \, \textnormal{d}\theta \, \textnormal{d}r \, ,
\end{split}
\end{equation}
where
\begin{equation}\label{NGAM}
\Gamma^{\tau} := - \mathscr{S} \gamma^0 \, \widetilde{\gamma}'^{\tau} = \left(\begin{array}{cccc}
r_+^{- 1} [\Delta + 4 M r] & - \textnormal{i} a \sin{(\theta)} & 0 & 0 \\
\textnormal{i} a \sin{(\theta)} & r_+ & 0 & 0 \\
0 & 0 & r_+ & \textnormal{i} a \sin{(\theta)} \\
0 & 0 & - \textnormal{i} a \sin{(\theta)} & r_+^{- 1} [\Delta + 4 M r]
\end{array}\right).
\end{equation}
We point out that in the above derivation, we have first applied the relation $\sqrt{\Sigma} \, \1_{\mathbb{C}^4} = \mathscr{P} \mathscr{P}^{\dagger}$, then the transformation law for the matrix $\mathscr{S}'$, namely $\mathscr{S} = \mathscr{P}^{\dagger} \mathscr{S}' \mathscr{P}$, and finally we have used the fact that both $\mathscr{S}$ and the product $\mathscr{P} \mathscr{S}$ are self-adjoint, which leads to $\mathscr{P} \mathscr{S} = \mathscr{S} \mathscr{P}^{\dagger}$. Besides, the integration limits are suppressed for ease of notation. The eigenvalues $\lambda_1, \lambda_2$ of the matrix (\ref{NGAM}) are positive 
\begin{equation*}
\begin{split}
& \lambda_1 = \frac{1}{2} \Biggl(r_+ + \frac{\Delta + 4 M r}{r_+} + \sqrt{\biggl(r_+ - \frac{\Delta + 4 M r}{r_+}\biggr)^2 + 4 a^2 \sin^2{(\theta)}} \,\, \Biggr) > 0 \\ \\
& \lambda_2 = \frac{1}{2} \Biggl(r_+ + \frac{\Delta + 4 M r}{r_+} - \sqrt{\biggl(r_+ + \frac{\Delta + 4 M r}{r_+}\biggr)^2 - 4 (\Sigma + 2 M r)} \,\, \Biggr) > 0 
\end{split}
\end{equation*}
and with algebraic multiplicities $\mu_A(\lambda_1) = \mu_A(\lambda_2) = 2$, demonstrating that (\ref{GSP3}) is indeed positive-definite. The symmetry property of the Dirac Hamiltonian (\ref{DHCF}) with respect to this scalar product is explicitly proven in Appendix \ref{appA}.

%------------------------------------------------------------------------------------
\section{Essential Self-adjointness of the Dirac Hamiltonian} \label{V}
%------------------------------------------------------------------------------------

\noindent In this section, we show that the Dirac Hamiltonian in the non-extreme Kerr geometry in horizon-penetrating advanced Eddington--Finkelstein-type coordinates is essentially self-adjoint using the results obtained in \cite{FRö} (we recently learned that in \cite{RauchTaylor} related results were found with different methods). Having an essentially self-adjoint Hamiltonian is necessary for the derivation of the integral representation of the Dirac propagator presented in the subsequent section, as it is based on the spectral theorem for unbounded, self-adjoint operators. The proof of the essential self-adjointness involves the technical difficulty that in the Kerr geometry the Dirac Hamiltonian is only almost everywhere elliptic, and hence not uniformly elliptic. More precisely, it fails to be elliptic at the event and the Cauchy horizon. This can be easily seen from the evaluation of the determinant of the principal symbol of the Hamiltonian (\ref{DHCF})
\begin{equation} \label{PrinSym}
P(r, \theta; \boldsymbol{\xi}) = \alpha^j(r, \theta) \, \xi_ j \, ,
\end{equation}
where $\boldsymbol{\xi} \in T^{\star} \mathfrak{N}_{\tau}$. In more detail, we first rewrite the matrices $\alpha^j$ in terms of the original Dirac matrices $\gamma^j$
\begin{equation*}
\alpha^j = - \textnormal{i} \, (\widetilde{\gamma}'^{\tau})^{- 1} \, \widetilde{\gamma}'^j = - \textnormal{i} \, (\gamma'^{\tau})^{- 1} \gamma'^j = - \textnormal{i} \, \mathscr{P} (\gamma^{\tau})^{- 1} \gamma^j \mathscr{P}^{-1} \, .
\end{equation*}
Substituting this expression into the principal symbol (\ref{PrinSym}) and computing the determinant yields
\begin{equation*}
\textnormal{det}\bigl(P(r, \theta; \boldsymbol{\xi})\bigr) = \frac{\textnormal{det}(\gamma^j \xi_j)}{\textnormal{det}(\gamma^{\tau})} \, .
\end{equation*}
Using the relations 
\begin{equation*}
(\gamma^{\tau})^2 = g^{\tau \tau} \1_{S_{\boldsymbol{x}}\mathfrak{M}} \quad \textnormal{and} \quad \gamma^i \xi_i \, \gamma^j \xi_j = g^{i j} \xi_i \, \xi_j \, \1_{S_{\boldsymbol{x}}\mathfrak{M}} \, ,
\end{equation*}
we obtain
\begin{equation} \label{detps}
\textnormal{det}\bigl(P(r, \theta; \boldsymbol{\xi})\bigr) = \biggl(\frac{g^{i j} \xi_i \,  \xi_j}{g^{\tau \tau}}\biggr)^2 .
\end{equation}
The Hamiltonian fails to be elliptic if its principal symbol is not invertible, that is, if the determinant (\ref{detps}) vanishes. This is the case for 
\begin{equation} \label{PSCOND}
g^{i j} \xi_i \xi_j = 0 \quad \textnormal{with} \quad \boldsymbol{\xi} \not= \boldsymbol{0} \, , 
\end{equation}
where the quantities $g^{i j}$ are the components of the inverse of the spacetime metric (\ref{Kerrmetric})  
\begin{equation*}
\begin{split}
\boldsymbol{g} & = \frac{1}{\Sigma} \, \bigl([\Sigma + 2 M r] \, \partial_{\tau} \otimes \partial_{\tau} - 2 M r \, (\partial_{\tau} \otimes \partial_{r} + \partial_{r} \otimes \partial_{\tau}) - \Delta \, \partial_{r} \otimes \partial_{r} \\ 
& \hspace{1.0cm} - a \, (\partial_{r} \otimes \partial_{\phi} + \partial_{\phi} \otimes \partial_{r}) - \partial_{\theta} \otimes \partial_{\theta} - \textnormal{csc}^2(\theta) \, \partial_{\phi} \otimes \partial_{\phi}\bigr) \, . 
\end{split}
\end{equation*}
Analyzing condition (\ref{PSCOND}), we find that the Hamiltonian is not elliptic at the event and the Cauchy horizon. We point out that by using the intrinsic Dirac Hamiltonian on the space-like hypersurface $\mathfrak{N}_{\tau}$, ellipticity would be conserved even at the horizons, which can be inferred from the analog condition
\begin{equation*} 
g^{i j}_{|\mathfrak{N}_{\tau}} \xi_i \xi_j = 0 \quad \textnormal{with} \quad \boldsymbol{\xi} \not= \boldsymbol{0} \, , 
\end{equation*}
where the $g^{i j}_{|\mathfrak{N}_{\tau}}$ denote the components of the inverse of the associated induced Riemannian metric 
\begin{equation*}
\boldsymbol{g}_{|\mathfrak{N}_{\tau}} = - \frac{1}{\Sigma} \Biggl(\frac{\bigl(r^2 + a^2\bigr)^2 - \Delta \, a^2 \sin^2{(\theta)}}{\Sigma + 2 M r} \, \partial_{r} \otimes \partial_{r} + a \, [\partial_{r} \otimes \partial_{\phi} + \partial_{\phi} \otimes \partial_{r}] + \partial_{\theta} \otimes \partial_{\theta} + \textnormal{csc}^2(\theta) \, \partial_{\phi} \otimes \partial_{\phi}\Biggr) \, .
\end{equation*}
However, as we work with the Hamiltonian obtained from the Dirac operator in the full Kerr spacetime, ellipticity is broken. Therefore, we cannot employ standard techniques and results from elliptic theory in order to verify the essential self-adjointness of the Dirac Hamiltonian. Instead, we apply the results derived in \cite{FRö}, where near-boundary elliptic methods are combined with results from the theory of symmetric hyperbolic systems (see, e.g., \cite{taylor1, John, Bartnik, Chernoff}). In the following, we state the geometrical and functional analytic settings for the formulation of the Cauchy problem for the massive Dirac equation in the non-extreme Kerr geometry in horizon-penetrating coordinates in Hamiltonian form, which is used as a technical tool in the proof of the essential self-adjointness of the Dirac Hamiltonian. 

%
%----------------------------------------------------------------
%
\begin{figure}[t] 
\begin{center} 
\includegraphics[scale=0.65]{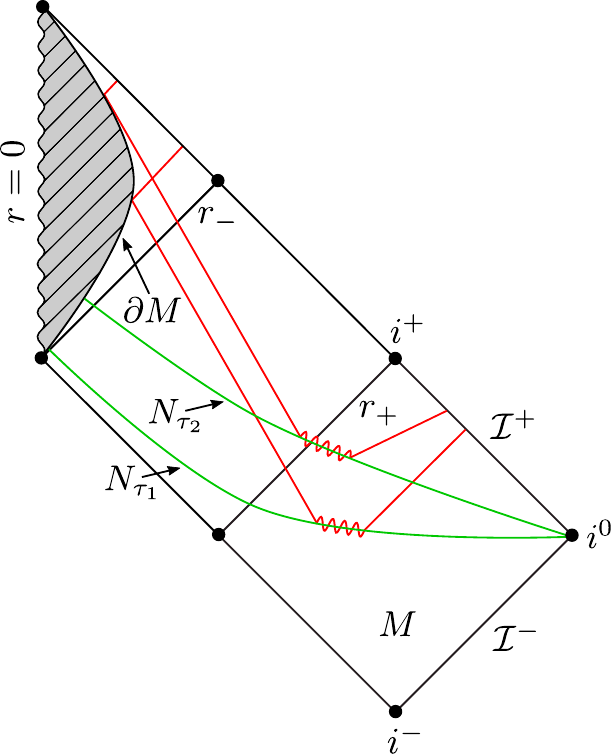}
\caption{\small{\label{PIC4} Carter--Penrose diagram of the region $M$ of the non-extreme Kerr geometry with constant-$\tau$ hypersurfaces $N_{\tau_1}$ and $N_{\tau_2}$ cut-off at the boundary $\partial M$. A radial Dirichlet-type boundary condition imposed on $\partial M$ leads to a reflection of Dirac particles away from the singularity without affecting their dynamics outside the Cauchy horizon. This is represented by Cauchy data propagated in $\tau$-direction.}}
\end{center}
\end{figure}
%
%----------------------------------------------------------------
%

We let $(\mathfrak{M}, \boldsymbol{g})$ be the non-extreme Kerr geometry with the metric (\ref{Kerrmetric}) in horizon-penetrating advanced Eddington--Finkelstein-type coordinates (\ref{AEFSTC}) and consider the subset 
\begin{equation*}
M := \{\tau, r > r_0, \theta, \phi\} \subset \mathfrak{M} \, , \quad \textnormal{where} \quad r_0 < r_- \, .
\end{equation*}
Furthermore, we introduce the time-like inner boundary 
\begin{equation*}
\partial M := \{\tau, r = r_0, \theta, \phi\} \quad \textnormal{of} \quad \mathfrak{M}
\end{equation*}
and the family of space-like hypersurfaces 
\begin{equation*}
N = (N_{\tau})_{\tau \in \mathbb{R}} \, , \quad \textnormal{where} \quad N_{\tau} := \{\tau = \textnormal{const.}, r > r_0, \theta, \phi\} \, ,
\end{equation*}
with boundaries 
\begin{equation*}
\partial N_{\tau} := \partial M \cap N_{\tau} \simeq S^2
\end{equation*}
(see FIG.\ \ref{PIC4}). These hypersurfaces constitute a foliation of $M$. Moreover, near $\partial M$, we define a locally time-like Killing vector field \cite{BON}
\begin{equation} \label{KKVF}
K := \partial_{\tau} + \beta_0 \, \partial_{\phi} \quad \textnormal{with} \quad \beta_0 = \beta_0(r_0) \in \mathbb{R} \backslash \{0\} \, ,
\end{equation}
where $\partial_{\tau}$ and $\partial_{\phi}$ are the Killing fields describing the stationarity and axisymmetry of the Kerr geometry, respectively. We note that the specific proof of existence of unique, global solutions of the Cauchy problem for the massive Dirac equation in Hamiltonian form for a general class of mixed initial-boundary value problems on Lorentzian manifolds presented in \cite{FRö} makes essential use of a Killing field $K = \partial_t$ that is tangential to and time-like on the inner boundary $\partial M$ and represented by a coordinate system describing an observer who is co-moving along the associated flow lines. This Killing field may also be space-like or null in $M\backslash\partial M$. A direct computation shows that the Killing field $K = \partial_{\tau}$ in advanced Eddington--Finkelstein-type coordinates is not everywhere time-like on $\partial M$. To be more precise, the condition for $\partial_{\tau}$ not being time-like reads
\begin{equation*}
\boldsymbol{g}(\partial_{\tau}, \partial_{\tau}) = 1 - \frac{2 M r}{\Sigma} \leq 0 \, .
\end{equation*}
This inequality is solved by
\begin{equation*}
M - \sqrt{M^2 - a^2 \cos^2{(\theta)}} \leq r \leq M + \sqrt{M^2 - a^2 \cos^2{(\theta)}} \, ,  
\end{equation*}
which is the ergosphere region. But taking $K$ as the linear combination (\ref{KKVF}), it turns out to be a Killing field that satisfies all the above assumptions. To make the connection between (\ref{KKVF}) and the Killing field $K = \partial_t$, we use the coordinate transformation 
\begin{equation*}
\mathbb{R} \times \mathbb{R}_{> 0} \times [0, \pi] \times [0, 2 \pi) \rightarrow \mathbb{R} \times \mathbb{R}_{> 0} \times [0, \pi] \times [0, 2 \pi) \, , \,\,\,\,\,\,\,\, (\tau, r, \theta, \phi) \mapsto (t, r, \theta, \Phi)
\end{equation*}
with  
\begin{equation} \label{tTrafo}
t = \tau \quad \textnormal{and} \quad \Phi = \phi - \beta_0 \tau \, .
\end{equation}
This transformation can be easily derived from the condition
\begin{equation*}
K = \partial_t =  \frac{\partial \tau}{\partial t} \, \partial_{\tau} + \frac{\partial \phi}{\partial t} \, \partial_{\phi} = \partial_{\tau} + \beta_0 \, \partial_{\phi} \, .
\end{equation*}
Evaluation of the gradient 
\begin{equation*}
\boldsymbol{\nabla} t = g^{\mu \nu} (\partial_{\mu} t) \, \partial_{\nu} = g^{t \nu} \, \partial_{\nu}
\end{equation*}
and subsequently
\begin{equation*}
\boldsymbol{g}(\partial_t, \boldsymbol{\nabla} t) = g^t_t = 1 > 0 \quad \textnormal{as well as} \quad \boldsymbol{g}(\boldsymbol{\nabla} t, \boldsymbol{\nabla} t) = g^{t t} = 1 + \frac{2 M r}{\Sigma} > 0
\end{equation*}
demonstrates that $\boldsymbol{\nabla} t$ is future-pointing and time-like and, hence, that the coordinate $t$ is a time function as is the original time coordinate $\tau$ \cite{R}. Due to the specific form of the transformation (\ref{tTrafo}), we find that the induced metric $\boldsymbol{g}_{| t = \textnormal{const.}}$ on the level sets of $t$ is identical to the induced metric 
\begin{equation*} 
\begin{split}
\boldsymbol{g}_{|\tau = \textnormal{const.}} = \boldsymbol{g}_{|\mathfrak{N}_{\tau}} & = - \biggl(1 + \frac{2 M r}{\Sigma}\biggr) \textnormal{d}r \otimes \textnormal{d}r + a \sin^2{(\theta)} \biggl(1 + \frac{2 M r}{\Sigma}\biggr) \bigl[\textnormal{d}r \otimes \textnormal{d}\phi + \textnormal{d}\phi \otimes \textnormal{d}r\bigr] \\ \\ 
& \hspace{0.4cm} - \Sigma \, \textnormal{d}\theta \otimes \textnormal{d}\theta - \sin^2{(\theta)} \biggl[\Sigma + a^2 \sin^2{(\theta)} \biggl(1 + \frac{2 M r}{\Sigma}\biggr)\biggr] \textnormal{d}\phi \otimes \textnormal{d}\phi
\end{split}
\end{equation*}
on the level sets of $\tau$ for the advanced Eddington--Finkelstein-type coordinates. As a consequence, all the results obtained for the co-moving coordinate system in \cite{FRö} will also hold true for the advanced Eddington--Finkelstein-type coordinates. 

Next, in addition to these geometric structures, we introduce the spin bundle $SM$ of $M$ with fibers $S_{\boldsymbol{x}}M \simeq \mathbb{C}^4$, where $\boldsymbol{x} \in M$. We may then consider the Dirac Hamiltonian given in (\ref{DHCF})
\begin{equation} \label{DIRH2}
H = \alpha^j \partial_j + \mathscr{V} \quad \textnormal{on} \quad N_{\tau} \, ,
\end{equation}
which is a symmetric operator with respect to the scalar product specified in (\ref{GSP3})
\begin{equation} \label{CSP}
(\psi' | \phi')_{|N_{\tau}} = \int_{r_0}^{\infty} \int_{S^2} \psi'^{\dagger} \, \Gamma^{\tau} \phi' \, \sin{(\theta)} \, \textnormal{d}\phi \, \textnormal{d}\theta \, \textnormal{d}r \, ,
\end{equation}
where the Dirac $4$-spinors $\psi', \phi' \in S_{\boldsymbol{x}}M$ and the matrix $\Gamma^{\tau}$ is defined in (\ref{NGAM}). We note that for $H$ being symmetric with respect to (\ref{CSP}), we have to impose the radial Dirichlet-type boundary condition (see Appendix \ref{appA})
\begin{equation*}
\bigl(\boldsymbol{\slashed{n}} - \textnormal{i} \mathscr{P} (\mathscr{P}^{\dagger})^{- 1}\bigr) \, \psi'_{|\partial M} = \boldsymbol{0}  \, ,
\end{equation*}
in which $\boldsymbol{n}$ is the inner normal to $\partial M$ and $\mathscr{P}$ is determined by (\ref{spinortr}). This boundary condition has the effect that Dirac particles are reflected at $\partial M$ away from the singularity such that, without changing their dynamics outside the Cauchy horizon, we obtain a unitary time evolution. The specific domain of the Hamiltonian reads
\begin{equation*}
\textnormal{Dom}(H) = \bigl\{\psi' \in C_0^{\infty}(N_{\tau}, SM) \,\, \big| \,\, \bigl(\boldsymbol{\slashed{n}} - \textnormal{i} \mathscr{P} (\mathscr{P}^{\dagger})^{- 1}\bigr) \psi'_{|\partial N_{\tau}} = \boldsymbol{0}\bigr\} \, .
\end{equation*}
In this setting, we find a unique, global solution of the Cauchy problem for the massive Dirac equation in Hamiltonian form in the class $C_{\textnormal{sc}}^{\infty}(M, SM)$.
\begin{Lemma}
The Cauchy problem for the massive Dirac equation in the non-extreme Kerr geometry in horizon-penetrating advanced Eddington--Finkelstein-type coordinates
\begin{equation*}
\left\{ \begin{array}{l} \hspace{0.05cm} \textnormal{i} \partial_{\tau} \psi' = H \psi' \\ [0.6em]
\psi'_{| \tau = 0} =: \psi'_0 \, \in \, C_0^{\infty}(N_{\tau = 0}, SM) \end{array} \right.
\end{equation*}
with the radial Dirichlet-type boundary condition at $\partial M$ given by
\begin{equation} \label{DTBCKERRGEOM}
\bigl(\boldsymbol{\slashed{n}} - \textnormal{i} \mathscr{P} (\mathscr{P}^{\dagger})^{- 1}\bigr) \, \psi'_{|\partial M} = \boldsymbol{0} \, ,
\end{equation}
where the initial data $\psi'_0$ is smooth, compactly supported outside, across, and inside the event horizon, up to the Cauchy horizon, and is compatible with the boundary condition, i.e.,
\begin{equation*}
\bigl(\boldsymbol{\slashed{n}} - \textnormal{i} \mathscr{P} (\mathscr{P}^{\dagger})^{- 1}\bigr) \, (H^p \psi'_0)_{|\partial N_{\tau}} = \boldsymbol{0} \quad \forall \quad p \in \mathbb{N}_0 \, ,
\end{equation*}
has a unique, global solution $\psi'$ in the class of smooth functions with spatially compact support $C^{\infty}_{\textnormal{sc}}(M, SM)$. Evaluating this solution at subsequent times $\tau$ and $\tau'$ gives rise to a unique unitary propagator 
\begin{equation*}
U^{\tau', \, \tau}: \, C_0^{\infty}(N_{\tau}, SM) \rightarrow C_0^{\infty}(N_{\tau'}, SM) \, .
\end{equation*} 
\end{Lemma}
\noindent The proof of this lemma is shown in detail for a more general class of non-uniformly elliptic mixed initial-boundary value problems for the Dirac equation in Hamiltonian form on Lorentzian manifolds with dimension $d \geq 3$ in \cite{FRö}. The existence of a unique, global solution $\psi' \in C^{\infty}_{\textnormal{sc}}(M, SM)$ is imperative for the specific proof of the essential self-adjointness of the Dirac Hamiltonian presented in the same work. Below, we state the result for the Hamiltonian in the non-extreme Kerr geometry in horizon-penetrating coordinates (\ref{DIRH2}).
\begin{Thm} \label{ThmIV2}
The massive Dirac Hamiltonian $H$ in the non-extreme Kerr geometry in horizon-penetrating advanced Eddington--Finkelstein-type coordinates with domain of definition 
\begin{equation*}
\textnormal{Dom}(H) = \bigl\{\psi' \in C_0^{\infty}(N_{\tau}, SM) \,\, \big| \,\, \bigl(\boldsymbol{\slashed{n}} - \textnormal{i} \mathscr{P} (\mathscr{P}^{\dagger})^{- 1}\bigr) \, (H^p \psi')_{|\partial N_{\tau}} = \boldsymbol{0} \,\,\,\, \forall \,\,\,\, p \in \mathbb{N}_0\bigr\}
\end{equation*}
is essentially self-adjoint. 
\end{Thm}

%--------------------------------------------------------------------------------------------------------------
\section{Resolvent of the Dirac Hamiltonian and Integral Spectral Representation of the Dirac Propagator} \label{VI}
%--------------------------------------------------------------------------------------------------------------

\noindent We may now construct an integral spectral representation of the Dirac propagator that yields the dynamics of massive Dirac particles outside, across, and inside the event horizon, up to the Cauchy horizon. More precisely, we derive an explicit expression for the spectral measure $\textnormal{d}P_{\omega}$ of the essentially self-adjoint Dirac Hamiltonian $H$ defined in (\ref{DIRH2}) with the domain specified in Theorem \ref{ThmIV2}, which arises in the formal spectral decomposition of the Dirac propagator 
\begin{equation} \label{sdr}
\psi'(\tau, r, \theta, \phi) = e^{- \textnormal{i} \tau H} \psi'_0(r, \theta, \phi) = \int_{\mathbb{R}} e^{- \textnormal{i} \omega \tau} \psi'_0(r, \theta, \phi) \, \textnormal{d}P_{\omega} \, ,
\end{equation}
where $\psi'_0$ is smooth, spatially compact initial data. To this end, we employ Stone's formula and, thus, express the spectral measure in terms of the resolvent $\textnormal{Res}(H) = (H - \omega_{\textnormal{c}})^{-1}$ of the Hamiltonian. As the spectrum of the Hamiltonian $\sigma(H) \subseteq \mathbb{R}$ is on the real line, this resolvent exists for all $\omega_{\textnormal{c}} \in \mathbb{C}\backslash\mathbb{R}$ with real part $\textnormal{Re}{(\omega_{\textnormal{c}})} = \omega \in \sigma(H)$ and is given uniquely. In the computation of the resolvent, we make use of quantities obtained in the analysis of Chandrasekhar's separation of variables, namely the angular projector onto a finite-dimensional, invariant eigenspace of the angular operator (\ref{angop}) and the Green's matrix of the radial ODE system (\ref{RADODESYS}).
\begin{Thm}\label{Prop}
The massive Dirac propagator in the non-extreme Kerr geometry in horizon-penetrating advanced Eddington--Finkelstein-type coordinates can be expressed via the integral spectral representation 
\begin{equation*}
\psi'(\tau, r, \theta, \phi) = \frac{1}{2 \pi \textnormal{i}} \, \sum_{k \in \mathbb{Z}} \, e^{- \textnormal{i} k \phi} \int_{\mathbb{R}} e^{- \textnormal{i} \omega \tau} \, \lim_{\epsilon \searrow 0} \, \bigl[(H_k - \omega - \textnormal{i} \epsilon)^{- 1} - (H_k - \omega + \textnormal{i} \epsilon)^{- 1}\bigr](r, \theta; r', \theta') \, \psi'_{0, k}(r', \theta') \, \textnormal{d}\omega \, ,
\end{equation*} 
where $\psi'_{0, k} \in C_0^{\infty}\bigl((r_0, \infty) \times [0, \pi], SM\bigr)$ is the initial data for fixed $k$-modes and $(H_k - \omega \mp \textnormal{i} \epsilon)^{- 1}$ are the resolvents of the Dirac Hamiltonian for fixed $k$-modes $H_k$ on the upper and lower complex half-planes. The resolvents are unique and of the form 
\begin{equation*}
\begin{split}
(H_k - \omega \mp \textnormal{i} \epsilon)^{- 1}(r, \theta; r', \theta') \, \psi'_{0, k}(r', \theta') & = - \sum_{l \in \mathbb{Z}} \int_{- 1}^1 Q_{l}(\theta; \theta') \int_{r_0}^{\infty} \mathscr{C} \left(\begin{array}{cc}
G(r; r')_{k, l, \omega \pm \textnormal{i} \epsilon} & \boldsymbol{0}_{\mathbb{C}^2} \\ 
\boldsymbol{0}_{\mathbb{C}^2} & G(r; r')_{k, l, \omega \pm \textnormal{i} \epsilon}
\end{array}\right) \\ \\ 
& \hspace{3.9cm} \times \mathscr{E}(r', \theta') \, \psi'_{0, k}(r', \theta') \, \textnormal{d}r' \textnormal{d}\bigl(\cos{(\theta')}\bigr)  
\end{split}
\end{equation*} 
with $Q_{l}(\, \cdot \, , \cdot \,)$ being the integral kernel of the spectral projector onto a finite-dimensional, invariant eigenspace of the angular operator (\ref{angop}) that corresponds to the spectral parameter $\xi_l$, $G(r; r')_{k, l, \omega \pm \textnormal{i} \epsilon}$ the two-dimensional Green's matrix of the radial first-order ODE system (\ref{RADODESYS}), and 
\begin{equation*}
\mathscr{C} = \left(\begin{array}{cccc}
1 & 0 & 0 & 0 \\ 
0 & 0 & 0 & 1 \\
0 & 1 & 0 & 0 \\
0 & 0 & 1 & 0 
\end{array}\right), \,\,\,\,\,
\mathscr{E}(r', \theta') = - 
\left(\begin{array}{cccc}
\textnormal{i} \bigl(\Delta(r') + 4 M r'\bigr) & r_+ a \sin{(\theta')} & 0 & 0 \\ 
0 & 0 & - \textnormal{i} r_+ & a \sin{(\theta')} \\
0 & 0 & r_+ a \sin{(\theta')} & \textnormal{i} \bigl(\Delta(r') + 4 M r'\bigr) \\
a \sin{(\theta')} & - \textnormal{i} r_+ & 0 & 0 
\end{array}\right).
\end{equation*} 
\end{Thm}

\begin{proof}
We first compute the resolvents $(H_k - \omega \mp \textnormal{i} \epsilon)^{- 1}$ of the Dirac Hamiltonian for fixed $k$-modes $H_k$, where $\omega \in \mathbb{R}$ and $\epsilon > 0$ is sufficiently small so that it can be considered as a slightly non-self-adjoint perturbation. For this purpose, we begin by substituting Chandrasekhar's mode ansatz with complex-valued frequencies	
\begin{equation*}
\psi'(\tau, r, \theta, \phi) =  e^{- \textnormal{i}(\omega_{\epsilon} \tau + k \phi)} \, \Psi(r, \theta) \,\,\,\,\,\,\,\, \textnormal{in which} \,\,\,\,\,\,\,\, \Psi \in L^2\bigl((r_0, \infty) \times [0, \pi], SM\bigr)
\end{equation*}
into the Dirac equation (\ref{DEHF}) restricted to $M$, yielding
\begin{equation}\label{SepHDE1} 
(H_k - \omega_{\epsilon}) \Psi = \boldsymbol{0} \, ,
\end{equation}
where $H_k := \alpha^r \partial_r + \alpha^{\theta} \partial_{\theta} - \textnormal{i} k \, \alpha^{\phi} + \mathscr{V}$ with the Dirac matrices $\alpha^j$, $j \in \{r, \theta, \phi\}$, and the potential $\mathscr{V}$ given in (\ref{Br})-(\ref{PotV}). We remark that we introduced the abbreviation $\omega_{\epsilon} \in \mathbb{C}$ with $\textnormal{Re}{(\omega_{\epsilon})} = \omega$ and $\textnormal{Im}{(\omega_{\epsilon})} \in \{- \epsilon, \epsilon\}$ in order to cover the resolvents in both the upper and lower complex half-planes simultaneously. Next, we define the spectral projector 
\begin{equation*} 
Q_{l} \Psi := \int_{- 1}^1 Q_{l}(\theta; \theta') \, \Psi(r, \theta') \, \textnormal{d}\bigl(\cos{(\theta')}\bigr)
\end{equation*}
onto the finite-dimensional, invariant eigenspace of the matrix-valued angular operator (\ref{angop}), which results from Chandrasekhar's separation of variables, corresponding to the spectral parameter $\xi_l$ with $l \in \mathbb{Z}$. This spectral projector is idempotent
\begin{equation*}
Q_{l}^n = Q_{l} \quad \textnormal{for all} \quad n \in \mathbb{N} \, ,
\end{equation*}
and the family of spectral projectors $(Q_l)_{l \in \mathbb{Z}}$ is complete
\begin{equation} \label{compconstr}
\sum_{l \in \mathbb{Z}} Q_{l} = \1 \, .
\end{equation}
We may therefore express the angular operator (\ref{angop}) by means of the family  $(Q_l)_{l \in \mathbb{Z}}$ as 
\begin{equation*} 
A = \sum_{l \in \mathbb{Z}} \xi_l Q_l \, .
\end{equation*}
Applying this representation of the angular operator and the completeness constraint (\ref{compconstr}) in Eq.\ (\ref{SepHDE1}), we obtain
\begin{equation} \label{SepHDE2}
- (\Sigma + 2 M r)^{- 1} \sum_{l \in \mathbb{Z}} \mathcal{M}(\partial_r; r, \theta)_{k, l, \omega_{\epsilon}} \, Q_{l} \, \Psi = \boldsymbol{0} \, ,
\end{equation}
where
\begin{equation*}
\mathcal{M}(\partial_r; r, \theta)_{k, l, \omega_{\epsilon}} := \left(\begin{array}{cccc}
\textnormal{i} \, O_{k, \omega_{\epsilon}} & a \sin{(\theta)} \, U_{\omega_{\epsilon}} & \textnormal{i} r_+ \, S_{l} & a \sin{(\theta)} \, \overline{S}_{l} \\ \\
\displaystyle \frac{a \sin{(\theta)}}{r_+} \, O_{k, \omega_{\epsilon}} & - \displaystyle \frac{\textnormal{i} (\Delta + 4 M r)}{r_+} \, U_{\omega_{\epsilon}} & a \sin{(\theta)} \, S_{l} & - \displaystyle \frac{\textnormal{i} (\Delta + 4 M r)}{r_+} \,  \overline{S}_{l} \\ \\ 
- \displaystyle \frac{\textnormal{i} (\Delta + 4 M r)}{r_+} \,  \overline{S}_{l} & a \sin{(\theta)} \, S_{l} & - \displaystyle \frac{\textnormal{i} (\Delta + 4 M r)}{r_+} \, U_{\omega_{\epsilon}} & \displaystyle \frac{a \sin{(\theta)}}{r_+} \, O_{k, \omega_{\epsilon}} \\ \\
a \sin{(\theta)} \, \overline{S}_{l} & \textnormal{i} r_+ \, S_{l} & a \sin{(\theta)} \, U_{\omega_{\epsilon}} & \textnormal{i} \, O_{k, \omega_{\epsilon}}  
\end{array}\right)
\end{equation*}
with the purely radial differential operators $O_{k, \omega_{\epsilon}}$, $U_{\omega_{\epsilon}}$, and the function $S_{l}$ defined by
\begin{equation*}
\begin{split}
O_{k, \omega_{\epsilon}} & := \Delta \, \partial_r + r - M - \textnormal{i} \omega_{\epsilon} (\Delta + 4 M r) - 2 \textnormal{i} a k \\ \\
U_{\omega_{\epsilon}} & := r_+ (\partial_r + \textnormal{i} \omega_{\epsilon}) \\ \\
S_{l} & := \xi_l + \textnormal{i} m r \, .
\end{split}
\end{equation*}
In the following, we show that the computation of the resolvent of the operator $\mathcal{M}(\partial_r; r, \theta)_{k, l, \omega_{\epsilon}}$ in (\ref{SepHDE2}) can be reduced to determining the two-dimensional Green's matrix of the radial ODE system (\ref{RADODESYS}). Writing the Dirac equation (\ref{SepHDE2}) in the factorized form
\begin{equation*}
- (\Sigma + 2 M r)^{- 1} \, \mathcal{B}(r, \theta) \sum_{l \in \mathbb{Z}} \mathcal{R}(\partial_r; r)_{k, l, \omega_{\epsilon}} \, Q_{l} \, \Psi = \boldsymbol{0} \, ,
\end{equation*}
where the matrix $\mathcal{B}(r, \theta)$ and the matrix-valued radial operator $\mathcal{R}(\partial_r; r)_{k, l, \omega_{\epsilon}}$ read
\begin{equation*}
\mathcal{B}(r, \theta) := 
\left(\begin{array}{cccc}
\textnormal{i} & a \sin{(\theta)} & 0 & 0 \\ 
\displaystyle \frac{a \sin{(\theta)}}{r_+} & - \displaystyle \frac{\textnormal{i} (\Delta + 4 M r)}{r_+} & 0 & 0 \\ 
0 & 0 & - \displaystyle \frac{\textnormal{i} (\Delta + 4 M r)}{r_+} & \displaystyle \frac{a \sin{(\theta)}}{r_+} \\ 
0 & 0 & a \sin{(\theta)} & \textnormal{i}  
\end{array}\right)
\end{equation*}
and
\begin{equation*}
\mathcal{R}(\partial_r; r)_{k, l, \omega_{\epsilon}} :=
\left(\begin{array}{cccc}
O_{k, \omega_{\epsilon}} & 0 & r_+ \, S_{l} & 0 \\ 
0 & U_{\omega_{\epsilon}} & 0 & \overline{S}_{l} \\ 
\overline{S}_{l} & 0 & U_{\omega_{\epsilon}} & 0 \\ 
0 & r_+ \, S_{l} & 0 & O_{k, \omega_{\epsilon}}  
\end{array}\right),
\end{equation*}
we can easily bring it into the block diagonal representation
\begin{equation} \label{SepHDE4}
- \mathscr{E}^{- 1}(r, \theta) \, \sum_{l \in \mathbb{Z}} \, \left(\begin{array}{cc}
\mathcal{R}^{2 \times 2}(\partial_r; r)_{k, l, \omega_{\epsilon}} & \boldsymbol{0}_{\mathbb{C}^2} \\ 
\boldsymbol{0}_{\mathbb{C}^2} & \mathcal{R}^{2 \times 2}(\partial_r; r)_{k, l, \omega_{\epsilon}}
\end{array}\right) \mathscr{C}^{- 1} \, Q_{l} \, \Psi = \boldsymbol{0} \, ,
\end{equation}
in which 
\begin{equation} \label{4t4mvro}
\left(\begin{array}{cc}
\mathcal{R}^{2 \times 2}(\partial_r; r)_{k, l, \omega_{\epsilon}} & \boldsymbol{0}_{\mathbb{C}^2} \\ 
\boldsymbol{0}_{\mathbb{C}^2} & \mathcal{R}^{2 \times 2}(\partial_r; r)_{k, l \omega_{\epsilon}}
\end{array}\right) = \mathscr{C}^{- 1} \, \mathcal{R}(\partial_r; r)_{k, l, \omega_{\epsilon}} \, \mathscr{C} 
\end{equation}
with
\begin{equation} \label{rad2t2op}
\mathcal{R}^{2 \times 2}(\partial_r; r)_{k, l, \omega_{\epsilon}} := 
\left(\begin{array}{cc}
O_{k, \omega_{\epsilon}} & r_+ \, S_{l} \\ 
\overline{S}_{l} & U_{\omega_{\epsilon}} 
\end{array}\right),
\end{equation}
and the matrices $\mathscr{E}^{- 1}(r, \theta)$ and $\mathscr{C}$ are defined by 
\begin{equation*}
\mathscr{E}^{- 1}(r, \theta) := (\Sigma + 2 M r)^{- 1} \, \mathcal{B}(r, \theta) \, \mathscr{C} 
\end{equation*}
and 
\begin{equation*}
\mathscr{C} := \left(\begin{array}{cccc}
1 & 0 & 0 & 0 \\ 
0 & 0 & 0 & 1 \\
0 & 1 & 0 & 0 \\
0 & 0 & 1 & 0 
\end{array}\right).
\end{equation*}
From the specific form of (\ref{SepHDE4}), it is obvious that the key quantity in the determination of the resolvent of $\mathcal{M}(\partial_r; r, \theta)_{k, l, \omega_{\epsilon}}$, and thus of the resolvent $(H_k - \omega_{\epsilon})^{- 1}$, is the solution $G(r; r')_{k, l, \omega_{\epsilon}}$ of the distributional equation
\begin{equation} \label{Green}
\mathcal{R}^{2 \times 2}(\partial_r; r)_{k, l, \omega_{\epsilon}} \, G(r; r')_{k, l, \omega_{\epsilon}} = \delta(r - r') \, \1_{\mathbb{C}^2} \, .
\end{equation}
We point out that (\ref{rad2t2op}) is identical to the operator (\ref{RADODESYS}), but for a complex-valued frequency $\omega$ with imaginary part $\textnormal{Im}(\omega) \in \{- \epsilon, \epsilon\}$. Hence, $G(r; r')_{k, l, \omega_{\epsilon}}$ corresponds to the Green's matrix of the radial ODE system obtained via Chandrasekhar's separation of variables. However, in the case of a complex-valued frequency, the solution of the radial ODE system has an additional dampening contribution guaranteeing that $\Psi(r, \theta)$ is in $L^2\bigl((r_0, \infty) \times [0, \pi], SM\bigr)$, which is in contrast to the original case with $\omega \in \mathbb{R}$. In order to solve the distributional equation (\ref{Green}), we first introduce the vector-valued functions 
\begin{equation*}
\Phi_1(r; r') = \left(\begin{array}{c}
\Phi_{1, 1}(r; r') \\ 
\Phi_{1, 2}(r; r')
\end{array}\right)
\quad \textnormal{and} \quad 
\Phi_2(r; r') = \left(\begin{array}{c}
\Phi_{2, 1}(r; r') \\ 
\Phi_{2, 2}(r; r') 
\end{array}\right)
\end{equation*}
that 
\begin{center}
\begin{itemize}
\item for $r \not= r'$ are linearly independent solutions of the homogeneous equation 
\begin{equation*}
\mathcal{R}^{2 \times 2}(\partial_r; r)_{k, l, \omega_{\epsilon}} \, \Phi(r; r') = \boldsymbol{0} \, ,
\end{equation*}
\item have jump discontinuities at $r = r'$,
\item satisfy the Dirichlet-type boundary condition (\ref{DTBCKERRGEOM}) at $r = r_0$\,,
\item and are square-integrable, that is
\begin{equation*}
\big\|\Phi_{1/2}(r; r')\big\|^2_2 = \int_{r_0}^{\infty} \big\|\Phi_{1/2}(r; r')\big\|^2 \, \textnormal{d}r < \infty \, .
\end{equation*}
\end{itemize}
\end{center}
These functions are specified in Appendix \ref{appB}. It turns out that their $r'$-dependence can be chosen in such a way that it is solely contained in Heaviside step functions $\Theta$. For clarity, these Heaviside step functions are explicitly stated in what follows, which makes it possible to consider $\Phi_1$ and $\Phi_2$ as functions of only the variable $r$. From the first and the last of the above properties as well as from Lemma \ref{L1} and Lemma \ref{L2a3} (but with a complex-valued frequency $\omega_{\epsilon} \in \{\omega + \textnormal{i} \epsilon, \omega - \textnormal{i} \epsilon\}$ and with the substitution $\sqrt{\omega^2 - m^2} \rightarrow \sqrt{|\omega_{\epsilon}|^2 - m^2}$), we can moreover infer that they have the specific asymptotics 
\begin{align*}
& \Phi_{1/2}(r) \sim  e^{\textnormal{i} \phi_+(r_{\star}(r))} \left(\begin{array}{cc} \displaystyle \frac{c_{1, \infty}}{\sqrt{\Delta}} \\ c_{2, \infty} \end{array}\right)
& & \textnormal{for} \quad r \rightarrow \infty \quad \textnormal{and} \quad \begin{cases} \textnormal{Im}{(\omega_{\epsilon})} < 0 \quad \textnormal{if} \quad |\omega_{\epsilon}| \geq m \\ \textnormal{Re}{(\omega_{\epsilon})} \geq 0 \, \quad \hspace{-0.05cm} \textnormal{if} \quad |\omega_{\epsilon}| < m \end{cases} \\ \nonumber \\
& \Phi_{1/2}(r) \sim  e^{- \textnormal{i} \phi_-(r_{\star}(r))} \left(\begin{array}{cc} \displaystyle \frac{c_{3, \infty}}{\sqrt{\Delta}} \\ c_{4, \infty} \end{array}\right) & & \textnormal{for} \quad r \rightarrow \infty \quad \textnormal{and} \quad \begin{cases} \textnormal{Im}{(\omega_{\epsilon})} > 0 \quad \textnormal{if} \quad |\omega_{\epsilon}| \geq m \\ \textnormal{Re}{(\omega_{\epsilon})} < 0 \, \quad \hspace{-0.05cm} \textnormal{if} \quad |\omega_{\epsilon}| < m \end{cases} \\ \nonumber \\
& \Phi_{1/2}(r) \sim \left(\begin{array}{cc}
\displaystyle \frac{c_{1, r_\pm}}{\sqrt{|\Delta|}} \, e^{2 \textnormal{i} \bigl(\omega_{\epsilon} + k \Omega^{(\pm)}_{\textnormal{Kerr}}\bigr) r_{\star}(r)} \\ 
c_{2, r_\pm}
\end{array}\right)
& & \textnormal{for} \quad \begin{cases} r \rightarrow r_+ \quad \textnormal{and} \quad \textnormal{Im}{(\omega_{\epsilon})} < 0 \\
r \rightarrow r_- \quad \textnormal{and} \quad \textnormal{Im}{(\omega_{\epsilon})} > 0 \end{cases} \\ \nonumber \\
& \Phi_{1/2}(r) \sim c_{3, r_\pm} \left(\begin{array}{cc} 0 \\ 1 \end{array}\right) & & \textnormal{for} \quad \begin{cases} r \rightarrow r_+ \quad \textnormal{and} \quad \textnormal{Im}{(\omega_{\epsilon})} > 0 \\
r \rightarrow r_- \quad \textnormal{and} \quad \textnormal{Im}{(\omega_{\epsilon})} < 0 \, , \end{cases} 
\end{align*}
where $c_{n, \infty}$ and $c_{m, r_{\pm}}$, with $n \in \{1, 2, 3, 4\}$ as well as $m \in \{1, 2, 3\}$, are scalar constants. We then use the ansatz 
\begin{equation} \label{Green2}
G(r; r')_{k, l, \omega_{\epsilon}} = \begin{cases}
\Theta(r - r') \, \Phi_1(r) \, P_1(r') + \Theta(r' - r) \, \Phi_2(r) \, P_2(r') & \,\, \textnormal{for} \,\,\,\, r_+ < r' < \infty \,\,\,\, \textnormal{and} \,\,\,\, r_0 \leq r' \leq r_- \\ 
\Theta(r - r') \, \Phi_1(r) \, P_1(r') + \Theta(r - r') \, \Phi_2(r) \, P_2(r') & \,\, \textnormal{for} \,\,\,\, r_- < r' \leq r_+ 
\end{cases}
\end{equation}
for the radial Green's matrix in case 
\begin{equation*}
|\omega_{\epsilon}| \geq m \,\,\,\,\, \textnormal{and} \,\,\,\,\, \textnormal{Im}{(\omega_{\epsilon})} < 0 \quad \textnormal{or} \quad |\omega_{\epsilon}| < m \,\,\,\,\, \textnormal{and} \,\,\,\,\, \textnormal{Re}{(\omega_{\epsilon})} \geq 0 \, ,
\end{equation*}
whereas in case 
\begin{equation*}
|\omega_{\epsilon}| \geq m \,\,\,\,\, \textnormal{and} \,\,\,\,\, \textnormal{Im}(\omega_{\epsilon}) > 0 \quad \textnormal{or} \quad |\omega_{\epsilon}| < m \,\,\,\,\, \textnormal{and} \,\,\,\,\, \textnormal{Re}(\omega_{\epsilon}) < 0 \, ,
\end{equation*}
we employ the ansatz
\begin{equation} \label{Green2b}
G(r; r')_{k, l, \omega_{\epsilon}} = \begin{cases}
\Theta(r - r') \, \Phi_1(r) \, P_1(r') + \Theta(r' - r) \, \Phi_2(r) \, P_2(r') & \,\, \textnormal{for} \,\,\,\, r_+ < r' < \infty \,\,\,\, \textnormal{and} \,\,\,\, r_0 \leq r' \leq r_- \\ 
\Theta(r' - r) \, \Phi_1(r) \, P_1(r') + \Theta(r' - r) \, \Phi_2(r) \, P_2(r') & \,\, \textnormal{for} \,\,\,\, r_- < r' \leq r_+ \, ,
\end{cases} 
\end{equation}
in which $P_1$ and $P_2$ are unknowns yet to be determined. Applying the radial operator (\ref{rad2t2op}) to (\ref{Green2}) and (\ref{Green2b}), we obtain
\begin{equation*}
\mathcal{R}^{2 \times 2}(\partial_r; r)_{k, l, \omega_{\epsilon}} \, G(r; r')_{k, l, \omega_{\epsilon}} = 
\left(\begin{array}{cc}
\Delta & 0 \\ 
0 & r_+ 
\end{array}\right) 
\delta(r - r') \begin{cases}
\bigl[\hspace{0.275cm} \Phi_1(r') P_1(r') \mp \Phi_2(r') P_2(r') \bigr] & \textnormal{for} \,\,\,\,  (\ref{Green2}) \\ 
\bigl[\pm \Phi_1(r') P_1(r') - \Phi_2(r') P_2(r') \bigr] & \textnormal{for} \,\,\,\, (\ref{Green2b}) \, .
\end{cases} 
\end{equation*}
Comparing these equations with (\ref{Green}) yields the two systems
\begin{equation*}
\left(\begin{array}{cc}
\Delta^{- 1} & 0 \\ 
0 & r_+^{- 1} 
\end{array}\right) = 
\begin{cases}
\hspace{0.275cm} \Phi_1(r') P_1(r') \mp \Phi_2(r') P_2(r') & \textnormal{for} \,\,\,\,  (\ref{Green2}) \\ 
\pm \Phi_1(r') P_1(r') - \Phi_2(r') P_2(r') & \textnormal{for} \,\,\,\, (\ref{Green2b}) \, .
\end{cases} 
\end{equation*}
Their solutions $P_{1/2}(r')$ read 
\begin{equation*} \label{GSol1}
P_{1, 1}(r') = \frac{\Phi_{2, 2}(r')}{\Delta(r') \, W(r')} \, , \,\,\,\, P_{1, 2}(r') = - \frac{\Phi_{2, 1}(r')}{r_+ W(r')} \, , \,\,\,\, P_{2, 1}(r') = \pm \frac{\Phi_{1, 2}(r')}{\Delta(r') \, W(r')} \, , \,\,\,\, P_{2, 2}(r') = \mp \frac{\Phi_{1, 1}(r')}{r_+ W(r')} 
\end{equation*}
for the ansatz (\ref{Green2}) and 
\begin{equation*} \label{GSol2}
P_{1, 1}(r') = \pm \frac{\Phi_{2, 2}(r')}{\Delta(r') \, W(r')} \, , \,\,\,\, P_{1, 2}(r') = \mp \frac{\Phi_{2, 1}(r')}{r_+ W(r')} \, , \,\,\,\, P_{2, 1}(r') = \frac{\Phi_{1, 2}(r')}{\Delta(r') \, W(r')} \, , \,\,\,\, P_{2, 2}(r') = - \frac{\Phi_{1, 1}(r')}{r_+ W(r')} 
\end{equation*}
for ansatz (\ref{Green2b}), where 
\begin{equation*}
W(r') = W(\Phi_1, \Phi_2)(r') := \Phi_{1, 1}(r') \, \Phi_{2, 2}(r') - \Phi_{1, 2}(r') \, \Phi_{2, 1}(r')
\end{equation*}
is the Wronskian. Substituting these expressions into (\ref{Green2}) and (\ref{Green2b}), respectively, leads in case $r_+ < r' < \infty$ or $r_0 \leq r' \leq r_-$ to the Green's matrix  
\begin{equation} \label{GM1}
\begin{split}
G(r; r')_{k, l, \omega_{\epsilon}} & = \frac{1}{W(r')} \left[\Theta(r - r') \left(\begin{array}{cc}
\displaystyle \frac{\Phi_{1, 1}(r) \Phi_{2, 2}(r')}{\Delta(r')} & \displaystyle - \frac{\Phi_{1, 1}(r) \Phi_{2, 1}(r')}{r_+} \\ \\
\displaystyle \frac{\Phi_{1, 2}(r) \Phi_{2, 2}(r')}{\Delta(r')} & \displaystyle - \frac{\Phi_{1, 2}(r) \Phi_{2, 1}(r')}{r_+} 
\end{array}\right)_{k, l, \omega_{\epsilon}} \right. \\ \\
& \hspace{1.9cm} \left. + \Theta(r' - r)
\left(\begin{array}{cc}
\displaystyle \frac{\Phi_{2, 1}(r) \Phi_{1, 2}(r')}{\Delta(r')} & \displaystyle - \frac{\Phi_{2, 1}(r) \Phi_{1, 1}(r')}{r_+} \\ \\
\displaystyle \frac{\Phi_{2, 2}(r) \Phi_{1, 2}(r')}{\Delta(r')} & \displaystyle - \frac{\Phi_{2, 2}(r) \Phi_{1, 1}(r')}{r_+} 
\end{array}\right)_{k, l, \omega_{\epsilon}} \, \right] 
\end{split}
\end{equation}
for both (\ref{Green2}) and (\ref{Green2b}), whereas in case $r_- < r' \leq r_+$ it leads to the Green's matrices
\begin{equation} \label{GM2}
\begin{split}
G(r; r')_{k, l, \omega_{\epsilon}} & = \frac{1}{W(r')} \left(\begin{array}{cc}
\displaystyle \frac{\Phi_{1, 1}(r) \Phi_{2, 2}(r') - \Phi_{2, 1}(r) \Phi_{1, 2}(r')}{\Delta(r')} & 
\displaystyle  \frac{\Phi_{2, 1}(r) \Phi_{1, 1}(r') - \Phi_{1, 1}(r) \Phi_{2, 1}(r')}{r_+} \\ \\
\displaystyle \frac{\Phi_{1, 2}(r) \Phi_{2, 2}(r') - \Phi_{2, 2}(r) \Phi_{1, 2}(r')}{\Delta(r')} & 
\displaystyle  \frac{\Phi_{2, 2}(r) \Phi_{1, 1}(r') - \Phi_{1, 2}(r) \Phi_{2, 1}(r')}{r_+}  
\end{array}\right)_{k, l, \omega_{\epsilon}} \\ \\
& \hspace{0.5cm} \times
\begin{cases}
\hspace{0.275cm} \Theta(r - r') & \textnormal{for} \,\,\,\,  (\ref{Green2}) \\ 
- \Theta(r' - r) & \textnormal{for} \,\,\,\, (\ref{Green2b}) \, .
\end{cases} 
\end{split}
\end{equation}
Subsequently, we may directly read off the resolvent of the Dirac Hamiltonian for fixed $k$-modes from the block-diagonalized representation (\ref{SepHDE4}). We thus find
\begin{equation} \label{resolvent}
(H_k - \omega_{\epsilon})^{- 1} \Psi = - \sum_{l \in \mathbb{Z}} Q_l \int_{r_0}^{\infty} \mathscr{C}  
\left(\begin{array}{cc}
G(r; r')_{k, l, \omega_{\epsilon}} & \boldsymbol{0}_{\mathbb{C}^2} \\ 
\boldsymbol{0}_{\mathbb{C}^2} & G(r; r')_{k, l, \omega_{\epsilon}}
\end{array}\right) \mathscr{E}(r', \theta) \, \Psi(r', \theta) \, \textnormal{d}r' \, ,
\end{equation} 
where the Green's matrix $G(r; r')_{k, l, \omega_{\epsilon}}$ is given in (\ref{GM1}) and (\ref{GM2}). To show that this expression is actually the desired resolvent, we verify the identity 
\begin{equation} \label{resolentid}
(H_k - \omega_{\epsilon}) (H_k - \omega_{\epsilon})^{- 1} \Psi = \Psi \, .
\end{equation}
Accordingly, applying the operator in (\ref{SepHDE4}) to (\ref{resolvent}), we obtain in a first step
\begin{equation*}
\begin{split}
& (H_k - \omega_{\epsilon}) (H_k - \omega_{\epsilon})^{- 1} \Psi =  \mathscr{E}^{- 1}(r, \theta) \sum_{l \in \mathbb{Z}} \left(\begin{array}{cc}
\mathcal{R}^{2 \times 2}(\partial_r; r)_{k, l, \omega_{\epsilon}} & \boldsymbol{0}_{\mathbb{C}^2} \\ 
\boldsymbol{0}_{\mathbb{C}^2} & \mathcal{R}^{2 \times 2}(\partial_r; r)_{k, l, \omega_{\epsilon}}
\end{array}\right) \mathscr{C}^{- 1} \int_{- 1}^1 Q_{l}(\theta; \theta') \Biggl\{\sum_{m \in \mathbb{Z}} \Biggr.\\ \\
& \Biggl. \int_{-1}^{1} Q_{m}(\theta'; \theta'') \int_{r_0}^{\infty} \, \mathscr{C} 
\left(\begin{array}{cc}
G(r; r')_{k, m, \omega_{\epsilon}} & \boldsymbol{0}_{\mathbb{C}^2} \\ 
\boldsymbol{0}_{\mathbb{C}^2} & G(r; r')_{k, m, \omega_{\epsilon}}
\end{array}\right) \mathscr{E}(r', \theta'') \, \Psi(r', \theta'') \, \textnormal{d}r' \, \textnormal{d}\bigl(\cos{(\theta'')}\bigr)\Biggr\} \textnormal{d}\bigl(\cos{(\theta')}\bigr) \, .
\end{split}
\end{equation*} 
Moving the integral kernel of the spectral projector $Q_l(\theta; \theta')$ into the $\theta''$-integral and taking into account that the spectral projectors are idempotent, i.e., their integral kernels satisfy the relation
\begin{equation*}
Q_{l}(\theta; \theta') \, Q_{m}(\theta'; \theta'') = \delta_{l m} \, \delta\bigl(\cos{(\theta)} - \cos{(\theta')}\bigr) \, Q_{m}(\theta; \theta'') \, ,
\end{equation*} 
we infer, after evaluating the $\theta'$-integral and the sum over all integers $m$, that
\begin{equation*}
\begin{split}
& (H_k - \omega_{\epsilon}) (H_k - \omega_{\epsilon})^{- 1} \Psi = \mathscr{E}^{- 1}(r, \theta) \sum_{l \in \mathbb{Z}} \left(\begin{array}{cc}
\mathcal{R}^{2 \times 2}(\partial_r; r)_{k, l, \omega_{\epsilon}} & \boldsymbol{0}_{\mathbb{C}^2} \\ 
\boldsymbol{0}_{\mathbb{C}^2} & \mathcal{R}^{2 \times 2}(\partial_r; r)_{k, l, \omega_{\epsilon}}
\end{array}\right) \mathscr{C}^{- 1} \\ \\
& \times \int_{-1}^{1} Q_{l}(\theta; \theta'') \int_{r_0}^{\infty} \, \mathscr{C} 
\left(\begin{array}{cc}
G(r; r')_{k, l, \omega_{\epsilon}} & \boldsymbol{0}_{\mathbb{C}^2} \\ 
\boldsymbol{0}_{\mathbb{C}^2} & G(r; r')_{k, l, \omega_{\epsilon}}
\end{array}\right) \mathscr{E}(r', \theta'') \, \Psi(r', \theta'') \, \textnormal{d}r' \, \textnormal{d}\bigl(\cos{(\theta'')}\bigr) \, .
\end{split}
\end{equation*} 
Next, we can also move the constant matrix $\mathscr{C}^{- 1}$ as well as the matrix-valued radial operator (\ref{4t4mvro}) into the $\theta''$- and the $r'$-integral. Employing (\ref{Green}) yields
\begin{equation*}
(H_k - \omega_{\epsilon}) (H_k - \omega_{\epsilon})^{- 1} \Psi = \mathscr{E}^{- 1}(r, \theta) \sum_{l \in \mathbb{Z}} Q_{l} \int_{r_0}^{\infty} \delta(r - r') \, \mathscr{E}(r', \theta) \, \Psi(r', \theta) \, \textnormal{d}r' \, .
\end{equation*} 
Solving the integral with respect to the variable $r'$ and substituting the completeness constraint for the spectral projectors (\ref{compconstr}), we immediately obtain the identity (\ref{resolentid}). 

Having established the explicit form of the resolvent $(H_k - \omega_{\epsilon})^{- 1}$ in (\ref{resolvent}), we continue deriving the integral spectral representation of the Dirac propagator. To this end, we express the Dirac spinor $\psi'$ at time $\tau$ in terms of the propagator $U^{\tau, 0} = e^{- \textnormal{i} \tau H}$ applied to smooth, spatially compact initial data $\psi'_0$ at time $\tau = 0$ and expand the initial data in terms of $k$-modes
\begin{equation} \label{proptoid}
\psi' = e^{- \textnormal{i} \tau H} \, \psi'_0 = \sum_{k \in \mathbb{Z}} \, e^{- \textnormal{i} k \phi} \, e^{- \textnormal{i} \tau H_k} \, \psi'_{0, k} \, .
\end{equation} 
We furthermore introduce the spectral projector of the Dirac Hamiltonian for fixed $k$-modes $H_k$ onto the interval $I \subset \mathbb{R}$
\begin{equation*}
P_I(H_k) := \chi_I(H_k) \, ,
\end{equation*}
where $\chi_I$ denotes the characteristic function
\begin{equation*}
\chi_I(H_k) := 
\begin{cases}
1 & \textnormal{for} \,\,\,\, \omega \in I \\ 
0 & \textnormal{for} \,\,\,\,  \omega \notin I 
\end{cases} 
\end{equation*}
with $\omega \in \sigma(H_k)$. Then, by making use of the identity relation
\begin{equation*}
P_{\, (- \infty, \infty)}(H_k) = \1 \, ,
\end{equation*}
we write (\ref{proptoid}) as 
\begin{equation*} 
\psi' =  \sum_{k \in \mathbb{Z}} \, e^{- \textnormal{i} k \phi} \, e^{- \textnormal{i} \tau H_k} \lim_{a \rightarrow \infty} P_{\, (- a, a)}(H_k) \psi'_{0, k} = \frac{1}{2} \, \sum_{k \in \mathbb{Z}} \, e^{- \textnormal{i} k \phi} \lim_{a \rightarrow \infty} e^{- \textnormal{i} \tau H_k} \bigl[P_{\, (- a, a)}(H_k) + P_{\, [- a, a]}(H_k)\bigr] \psi'_{0, k} \, .
\end{equation*} 
Employing Stone's formula for the spectral projector of an unbounded, self-adjoint operator \cite{ReedSimon}, which in our framework reads 
\begin{equation*}
e^{- \textnormal{i} \tau H_k} \bigl[P_{\, (- a, a)}(H_k) + P_{\, [- a, a]}(H_k)\bigr] \psi'_{0, k} = \lim_{\epsilon \searrow 0} \frac{1}{\pi \textnormal{i}} \, \int_{- a}^{a} e^{- \textnormal{i} \omega \tau} \bigl[(H_k - \omega - \textnormal{i} \epsilon)^{- 1} - (H_k - \omega + \textnormal{i} \epsilon)^{- 1}\bigr] \psi'_{0, k} \, \textnormal{d}\omega \, ,
\end{equation*} 
we obtain
\begin{equation*} 
\psi' = \frac{1}{2 \pi \textnormal{i}} \, \sum_{k \in \mathbb{Z}} \, e^{- \textnormal{i} k \phi} \lim_{a \rightarrow \infty} \, \lim_{\epsilon \searrow 0} \, \int_{- a}^a e^{- \textnormal{i} \omega \tau} \bigl[(H_k - \omega - \textnormal{i} \epsilon)^{- 1} - (H_k - \omega + \textnormal{i} \epsilon)^{- 1}\bigr] \psi'_{0, k} \, \textnormal{d}\omega \, ,
\end{equation*}
where the resolvents are given by (\ref{resolvent}). Finally, since the fundamental solutions that occur in the resolvents are bounded for all $\epsilon > 0$ and all $\omega \in \mathbb{R}$ as shown in Appendix \ref{appB}, we can apply Lebesgue's dominated convergence theorem and commute the $\epsilon$-limit and the integral with respect to $\omega$, yielding
\begin{equation} \label{DirProp} 
\psi'(\tau, r, \theta, \phi) = \frac{1}{2 \pi \textnormal{i}} \, \sum_{k \in \mathbb{Z}} \, e^{- \textnormal{i} k \phi} \int_{\mathbb{R}} e^{- \textnormal{i} \omega \tau} \lim_{\epsilon \searrow 0} \, \bigl[(H_k - \omega - \textnormal{i} \epsilon)^{- 1} - (H_k - \omega + \textnormal{i} \epsilon)^{- 1}\bigr](r, \theta; r', \theta') \, \psi'_{0, k}(r', \theta') \, \textnormal{d}\omega \, .
\end{equation}
We note that by comparing this expression to formula (\ref{sdr}), we may directly identify the spectral measure $\textnormal{d}P_{\omega}$ of the Dirac Hamiltonian $H$.
\end{proof}

\noindent This integral spectral representation can be further simplified, on the one hand, by performing the limit $r_0 \nearrow r_-$ and, on the other hand, by computing the difference of the two resolvents $(H_k - \omega - \textnormal{i} \epsilon)^{- 1}$ and $(H_k - \omega + \textnormal{i} \epsilon)^{- 1}$ for $\epsilon \searrow 0$. In the following, we explicitly work out the case $|\omega_{\epsilon}| \geq m$. The case $|\omega_{\epsilon}| < m$ may be treated similarly. As the fundamental solutions $\Phi_1(r; r')$ and $\Phi_2(r; r')$, which constitute the radial Green's matrix $G(r; r')_{k, l, \omega_{\epsilon}}$ and therefore the resolvent $(H_k - \omega_{\epsilon})^{- 1}$, are given piecewise for the domains $r_+ < r' < \infty$, $r_- < r' \leq r_+$, and $r_0 \leq r' \leq r_-$ (see the second part of Appendix \ref{appB}), we begin by splitting the $r'$-integral in the difference of resolvents in the limit $\epsilon \searrow 0$ into the three associated contributions
\begin{equation} \label{simp1}
\begin{split}
& \lim_{\epsilon \searrow 0} \, \bigl[(H_k - \omega - \textnormal{i} \epsilon)^{- 1} - (H_k - \omega + \textnormal{i} \epsilon)^{- 1}\bigr] \psi'_{0, k} = \lim_{\epsilon \searrow 0} \, \sum_{l \in \mathbb{Z}} \, Q_l \biggl(\int_{r_0}^{r_-} + \int_{r_-}^{r_+} + \int_{r_+}^{\infty}\biggr) \mathscr{C} \\ \\
& \times 
\left(\begin{array}{cc}
G(r; r')_{k, l, \omega - \textnormal{i} \epsilon} - G(r; r')_{k, l, \omega + \textnormal{i} \epsilon} & \boldsymbol{0}_{\mathbb{C}^2} \\ 
\boldsymbol{0}_{\mathbb{C}^2} & G(r; r')_{k, l, \omega - \textnormal{i} \epsilon} - G(r; r')_{k, l, \omega + \textnormal{i} \epsilon}
\end{array}\right) \mathscr{E}(r', \theta) \, \psi'_{0, k}(r', \theta) \, \textnormal{d}r' \, .
\end{split}
\end{equation}
Because the integrands in (\ref{simp1}), and hence the $l$th summand, are bounded for all values of $\epsilon$, $r'$, and $\theta$ (see the first part of Appendix \ref{appB} and keeping in mind that the initial data for fixed $k$-modes $\psi'_{0, k}$ has spatially compact support), we can again employ Lebesgue's dominated convergence theorem, which allows us to commute the limit $\epsilon \searrow 0$ with the sum over the integers $l$, the spectral projector $Q_l$, and the integrals with respect to $r'$. Then, applying the limit $r_0 \nearrow r_-$ to the integral spectral representation (\ref{DirProp}) and commuting this limit with the sum over the integers $k$, the integral with respect to $\omega$, and last the sum over the integers $l$ as well as the spectral projector $Q_l$ (in the difference of resolvents) using the same reasoning as before, we obtain the expression
\begin{equation} \label{lldr}
\begin{split}
& \lim_{r_0 \nearrow r_-} \, \lim_{\epsilon \searrow 0} \, \bigl[(H_k - \omega - \textnormal{i} \epsilon)^{- 1} - (H_k - \omega + \textnormal{i} \epsilon)^{- 1}\bigr] \psi'_{0, k} = \sum_{l \in \mathbb{Z}} \, Q_l \biggl(\int_{r_-}^{r_+} + \int_{r_+}^{\infty}\biggr) \mathscr{C} \times \\ \\
& \left(\begin{array}{cc}
\displaystyle \lim_{\epsilon \searrow 0} \, \bigl[G(r; r')_{k, l, \omega - \textnormal{i} \epsilon} - G(r; r')_{k, l, \omega + \textnormal{i} \epsilon}\bigr] & \boldsymbol{0}_{\mathbb{C}^2} \\ 
\boldsymbol{0}_{\mathbb{C}^2} & \displaystyle \lim_{\epsilon \searrow 0} \, \bigl[G(r; r')_{k, l, \omega - \textnormal{i} \epsilon} - G(r; r')_{k, l, \omega + \textnormal{i} \epsilon}\bigr]
\end{array}\right) \hspace{-0.1cm} \mathscr{E}(r', \theta) \, \psi'_{0, k}(r', \theta) \, \textnormal{d}r' \, .
\end{split}
\end{equation}
In order to compute the limit $\epsilon \searrow 0$ of the difference of the radial Green's matrices $G(r; r')_{k, l, \omega - \textnormal{i} \epsilon}$ and $G(r; r')_{k, l, \omega + \textnormal{i} \epsilon}$ in the domain $r_+ < r' < \infty$, we introduce the auxiliary functions (see the second part of Appendix \ref{appB})
\begin{equation*}
\chi_1(r) := \lim_{\epsilon \searrow 0} \widecheck{\Phi}^{(\infty)}(r) \quad \textnormal{and} \quad \chi_2(r) := \lim_{\epsilon \searrow 0} \widehat{\Phi}^{(\infty)}(r) \, ,
\end{equation*}
and write the $\epsilon$-limits of the fundamental radial solutions $\Phi_1$ and $\Phi_2$ as
\begin{equation} \label{frs1}
\begin{split}
\lim_{\epsilon \searrow 0} \Phi_1 & = \chi_1 \quad \textnormal{and} \quad \lim_{\epsilon \searrow 0} \Phi_2 = \alpha \chi_1 + \beta \chi_2 \quad \textnormal{for} \quad \textnormal{Im}{(\omega_{\epsilon})} > 0 \\ 
\lim_{\epsilon \searrow 0} \Phi_1 & = \chi_2 \quad \textnormal{and} \quad \lim_{\epsilon \searrow 0} \Phi_2 = \gamma \chi_1 + \delta \chi_2 \hspace{0.45cm} \textnormal{for} \quad \textnormal{Im}{(\omega_{\epsilon})} < 0 \, ,
\end{split}
\end{equation}
where $\alpha$, $\beta$, $\gamma$, and $\delta$ are constants. The corresponding Wronskian yields
\begin{equation} \label{Wr1}
\lim_{\epsilon \searrow 0} W(\Phi_1, \Phi_2) = \begin{cases}
\hspace{0.26cm} \beta \, W(\chi_1, \chi_2) & \textnormal{for} \quad  \textnormal{Im}{(\omega_{\epsilon})} > 0 \\ 
- \gamma \, W(\chi_1, \chi_2) & \textnormal{for} \quad \textnormal{Im}{(\omega_{\epsilon})} < 0 \, .
\end{cases} 
\end{equation}
Substitution of (\ref{frs1}) and (\ref{Wr1}) into (\ref{GM1}) results in
\begin{equation} \label{DGM1}
\begin{split}
\lim_{\epsilon \searrow 0} \, \bigl[G(r; r')_{k, l, \omega - \textnormal{i} \epsilon} &- G(r; r')_{k, l, \omega + \textnormal{i} \epsilon}\bigr]_{|r_+ < r' < \infty} \\ 
& = \frac{1}{W(\chi_1, \chi_2)(r')} \sum_{u, v = 1}^2 T_{u, v} \left(\begin{array}{ccc}
\displaystyle - \frac{\chi_{u, 1}(r) \, \chi_{v, 2}(r')}{\Delta(r')} & & \displaystyle \frac{\chi_{u, 1}(r) \, \chi_{v, 1}(r')}{r_+} \\ \\
\displaystyle - \frac{\chi_{u, 2}(r) \, \chi_{v, 2}(r')}{\Delta(r')} & & \displaystyle \frac{\chi_{u, 2}(r) \, \chi_{v, 1}(r')}{r_+}
\end{array}\right)_{k, l, \omega}   
\end{split}
\end{equation}
with the coefficients 
\begin{equation*}
T_{1, 1} = \frac{\alpha}{\beta} \, , \,\,\,\,\,\, T_{1, 2} = T_{2, 1} = 1 \, , \,\,\,\,\, \textnormal{and} \,\,\,\,\,\, T_{2, 2} = \frac{\delta}{\gamma} \, .
\end{equation*}
For the domain $r_- < r' \leq r_+$ on the other hand, we define the auxiliary functions 
\begin{equation*}
\widetilde{\chi}_1(r) := \Theta(r_+ - r) \lim_{\epsilon \searrow 0} \widehat{\Phi}^{(-)}(r) + \Theta(r - r_+) \lim_{\epsilon \searrow 0} \widehat{\Phi}^{(\infty)}(r) \quad \textnormal{and} \quad \widetilde{\chi}_2(r) := \Theta(r - r_-) \lim_{\epsilon \searrow 0} \widecheck{\Phi}^{(-)}(r) \, ,
\end{equation*}
and express the $\epsilon$-limits of the fundamental radial solutions by
\begin{equation} \label{frs2}
\begin{split}
\lim_{\epsilon \searrow 0} \Phi_1 & = \alpha' \widetilde{\chi}_1 + \beta' \widetilde{\chi}_2 \quad \textnormal{and} \quad \lim_{\epsilon \searrow 0} \Phi_2 = \widetilde{\chi}_2 \quad \textnormal{for} \quad \textnormal{Im}{(\omega_{\epsilon})} > 0 \\ 
\lim_{\epsilon \searrow 0} \Phi_1 & = \widetilde{\chi}_1 \quad \textnormal{and} \quad \lim_{\epsilon \searrow 0} \Phi_2 = \gamma' \widetilde{\chi}_1 + \delta' \widetilde{\chi}_2 \hspace{0.44cm} \textnormal{for} \quad \textnormal{Im}{(\omega_{\epsilon})} < 0 \, ,
\end{split}
\end{equation}
in which $\alpha'$, $\beta'$, $\gamma'$, and $\delta'$ are also constants. In this case, the Wronskian becomes
\begin{equation} \label{Wr2}
\lim_{\epsilon \searrow 0} W(\Phi_1, \Phi_2) = \begin{cases}
\alpha' \, W(\widetilde{\chi}_1, \widetilde{\chi}_2) & \textnormal{for} \quad  \textnormal{Im}{(\omega_{\epsilon})} > 0 \\ 
\delta' \, W(\widetilde{\chi}_1, \widetilde{\chi}_2) & \textnormal{for} \quad \textnormal{Im}{(\omega_{\epsilon})} < 0 \, .
\end{cases} 
\end{equation}
Using (\ref{frs2}) and (\ref{Wr2}) in (\ref{GM2}) to calculate the difference of Green's matrices for $\epsilon \searrow 0$ leads to
\begin{equation} \label{DGM2}
\begin{split}
\lim_{\epsilon \searrow 0} \, \bigl[G(r; r')_{k, l, \omega - \textnormal{i} \epsilon} &- G(r; r')_{k, l, \omega + \textnormal{i} \epsilon}\bigr]_{|r_- < r' \leq r_+} \\
& = \frac{1}{W(\widetilde{\chi}_1, \widetilde{\chi}_2)(r')} \sum_{u, v = 1}^2 \widetilde{T}_{u, v} \left(\begin{array}{ccc}
\displaystyle \frac{\widetilde{\chi}_{u, 1}(r) \, \widetilde{\chi}_{v, 2}(r')}{\Delta(r')} & & \displaystyle - \frac{\widetilde{\chi}_{u, 1}(r) \, \widetilde{\chi}_{v, 1}(r')}{r_+} \\ \\
\displaystyle \frac{\widetilde{\chi}_{u, 2}(r) \, \widetilde{\chi}_{v, 2}(r')}{\Delta(r')} & & \displaystyle - \frac{\widetilde{\chi}_{u, 2}(r) \, \widetilde{\chi}_{v, 1}(r')}{r_+}
\end{array}\right)_{k, l, \omega}
\end{split}
\end{equation}
with the coefficients 
\begin{equation*}
\widetilde{T}_{1, 1} = \widetilde{T}_{2, 2} = 0 \quad \textnormal{and} \quad \widetilde{T}_{1, 2} = - \widetilde{T}_{2, 1} = 1 \, .
\end{equation*}
Abbreviating (\ref{DGM1}) and (\ref{DGM2}) by $\mathcal{G}(r; r_+ < r' < \infty)_{k, l, \omega}$ and $\mathcal{G}(r; r_- < r' \leq r_+)_{k, l, \omega}$, respectively, and inserting these quantities into (\ref{lldr}), the Dirac propagator (\ref{DirProp}) yields 
\begin{equation*}
\begin{split}
\psi'(\tau, r, \theta, \phi) & = \frac{1}{2 \pi \textnormal{i}} \, \sum_{k, l \in \mathbb{Z}} \, e^{- \textnormal{i} k \phi} \int_{\mathbb{R}} e^{- \textnormal{i} \omega \tau} \, Q_{l} \, \mathscr{C} \biggl[\int_{r_-}^{r_+} \left(\begin{array}{cc}
\mathcal{G}(r; r_- < r' \leq r_+)_{k, l, \omega} & \boldsymbol{0}_{\mathbb{C}^2} \\ 
\boldsymbol{0}_{\mathbb{C}^2} & \mathcal{G}(r; r_- < r' \leq r_+)_{k, l, \omega}
\end{array}\right) \biggr. \\ \\
& \hspace{0.4cm} \biggl. + \int_{r_+}^{\infty} \left(\begin{array}{cc}
\mathcal{G}(r; r_+ < r' < \infty)_{k, l, \omega} & \boldsymbol{0}_{\mathbb{C}^2} \\ 
\boldsymbol{0}_{\mathbb{C}^2} & \mathcal{G}(r; r_+ < r' < \infty)_{k, l, \omega}
\end{array}\right)\biggr] \mathscr{E}(r', \theta) \, \psi'_{0, k}(r', \theta) \, \textnormal{d}r' \, \textnormal{d}\omega \\ \\
& = \frac{1}{2 \pi \textnormal{i}} \, \sum_{k, l \in \mathbb{Z}} \, e^{- \textnormal{i} k \phi} \int_{\mathbb{R}} e^{- \textnormal{i} \omega \tau} \, Q_{l} \, \mathscr{C} \, \biggl(\1_{\mathbb{C}^2} \otimes \, \biggl[\int_{r_-}^{r_+} \mathcal{G}(r; r_- < r' \leq r_+)_{k, l, \omega} \biggr. \biggr.\\ \\
& \hspace{0.4cm} \biggl.\biggl. + \int_{r_+}^{\infty} \mathcal{G}(r; r_+ < r' < \infty)_{k, l, \omega}\biggr]\biggr) \, \mathscr{E}(r', \theta) \, \psi'_{0, k}(r', \theta) \, \textnormal{d}r' \, \textnormal{d}\omega \, . 
\end{split}
\end{equation*}
Given in this form, the generalized, horizon-penetrating integral spectral representation of the massive Dirac propagator in the non-extreme Kerr geometry resembles the one restricted to the region outside the event horizon derived in \cite{FKSM3}.

\vspace{0.2cm}

\section*{Acknowledgments}
\noindent The authors are grateful to Niky Kamran, Guillaume Idelon--Riton, and Simone Murro for useful discussions and comments. This work was supported by the DFG research grant ``Dirac Waves in the Kerr Geometry: Integral Representations, Mass Oscillation Property and the Hawking Effect.''

\vspace{0.5cm}

\begin{appendix}

%------------------------------------------------------------------------------------------------
\section{Symmetry of the Dirac Hamiltonian and Dirichlet-type Boundary Condition} \label{appA}
%------------------------------------------------------------------------------------------------

\noindent In this appendix, we show the symmetry of the Dirac Hamiltonian $H$ with respect to the canonical scalar product $(\, \cdot \, | \, \cdot \, )_{| \mathfrak{N}_{\tau}}$ on the space-like hypersurface $\mathfrak{N}_{\tau}$ by direct computation. Furthermore, we introduce and discuss the relevant radial Dirichlet-type boundary condition imposed on the Dirac spinors.

\begin{Thm}
The Dirac Hamiltonian (\ref{DHCF}) is symmetric with respect to the scalar product (\ref{GSP3}). 
\end{Thm}

\begin{proof}
To establish the symmetry, namely that 
\begin{equation*}
(\psi' | H \phi')_{|\mathfrak{N}_{\tau}} = (H \psi' | \phi')_{|\mathfrak{N}_{\tau}} \, ,
\end{equation*}
we begin by splitting the potential $\mathscr{V}$ given in (\ref{PotV}) into mass-independent and mass-dependent parts
\begin{equation*}
\mathscr{V} = \mathscr{V}_0 + \mathscr{V}_m \, ,
\end{equation*}
where
\begin{equation*}
\mathscr{V}_0 := - \frac{1}{\Sigma + 2 M r}
\left(\begin{array}{cc}
\mathscr{B}_1 & \boldsymbol{0}_{\mathbb{C}^2} \\
\boldsymbol{0}_{\mathbb{C}^2} & \mathscr{B}_4
\end{array}\right)
\quad \textnormal{and} \quad 
\mathscr{V}_m := - \frac{1}{\Sigma + 2 M r}
\left(\begin{array}{cc}
\boldsymbol{0}_{\mathbb{C}^2} & \mathscr{B}_2 \\
\mathscr{B}_3 & \boldsymbol{0}_{\mathbb{C}^2}
\end{array}\right) .
\end{equation*}
The $(2 \times 2)$-blocks $\mathscr{B}_k$, with $k \in \{1, 2, 3, 4\}$, are specified in (\ref{B1234}). This procedure bears the advantage of obtaining anti-self-adjoint and self-adjoint matrices  
\begin{equation}\label{relprop}
\Gamma^{\tau} \mathscr{V}_0 = - \mathscr{V}_0^{\dagger} \Gamma^{\tau} \quad \textnormal{and} \quad \Gamma^{\tau} \mathscr{V}_m = \mathscr{V}_m^{\dagger} \Gamma^{\tau} \, ,
\end{equation}
for which $\Gamma^{\tau} = \Gamma^{\tau \dagger}$ is defined in (\ref{NGAM}). We may then write
\begin{equation*}
\begin{split}
(\psi' | H \phi')_{|\mathfrak{N}_{\tau}} & = \iiint \psi'^{\dagger} \, \Gamma^{\tau} H \, \phi' \sin{(\theta)} \,  \textnormal{d}\phi \, \textnormal{d}\theta \, \textnormal{d}r \\ \\
& = \iiint \psi'^{\dagger} \, \Gamma^{\tau} \alpha^j \partial_j(\phi') \sin{(\theta)} \, \textnormal{d}\phi \, \textnormal{d}\theta \, \textnormal{d}r + \iiint \psi'^{\dagger} \, \Gamma^{\tau} \mathscr{V}_0 \, \phi' \sin{(\theta)} \, \textnormal{d}\phi \, \textnormal{d}\theta \, \textnormal{d}r \\ \\
& \hspace{0.4cm} + \iiint \psi'^{\dagger} \, \Gamma^{\tau} \mathscr{V}_m \, \phi' \sin{(\theta)} \, \textnormal{d}\phi \, \textnormal{d}\theta \, \textnormal{d}r \, . 
\end{split}
\end{equation*}
Integration by parts of the first triple integral in the second line and substitution of the relations (\ref{relprop}) in the remaining two triple integrals results in 
\begin{equation}\label{PRSYM}
\begin{split}
(\psi' | H \phi')_{|\mathfrak{N}_{\tau}} & = - \iiint \partial_j \bigl(\psi'^{\dagger} \, \Gamma^{\tau} \alpha^j  \sin{(\theta)}\bigr) \phi' \, \textnormal{d}\phi \, \textnormal{d}\theta \, \textnormal{d}r - \iiint \psi'^{\dagger} \,  \mathscr{V}_0^{\dagger} \, \Gamma^{\tau} \phi' \sin{(\theta)} \, \textnormal{d}\phi \, \textnormal{d}\theta \, \textnormal{d}r \\ \\
& \hspace{0.4cm} + \iiint \psi'^{\dagger} \, \mathscr{V}_m^{\dagger} \, \Gamma^{\tau} \phi' \sin{(\theta)} \, \textnormal{d}\phi \, \textnormal{d}\theta \, \textnormal{d}r \\ \\
& = - \iiint \partial_j (\psi'^{\dagger}) \, \Gamma^{\tau} \alpha^j \phi' \sin{(\theta)} \, \textnormal{d}\phi \, \textnormal{d}\theta \, \textnormal{d}r \\ \\
& \hspace{0.4cm} - \iiint \psi'^{\dagger} \bigl[\partial_j (\Gamma^{\tau}) \, \alpha^j + \Gamma^{\tau} \partial_j (\alpha^j) + \Gamma^{\tau} \alpha^{\theta} \cot{(\theta)}\bigr] \phi' \sin{(\theta)} \, \textnormal{d}\phi \, \textnormal{d}\theta \, \textnormal{d}r \\ \\
& \hspace{0.4cm} - \iiint (\mathscr{V}_0 \, \psi')^{\dagger} \, \Gamma^{\tau} \phi' \sin{(\theta)} \, \textnormal{d}\phi \, \textnormal{d}\theta \, \textnormal{d}r + \iiint (\mathscr{V}_m \, \psi')^{\dagger} \, \Gamma^{\tau} \phi' \sin{(\theta)} \, \textnormal{d}\phi \, \textnormal{d}\theta \, \textnormal{d}r \, . 
\end{split}
\end{equation}
We remark that in the integration by parts, the angular derivatives do not give rise to boundary terms because the two-dimensional submanifold $S^2$ in $\mathfrak{N}_{\tau} \simeq \mathbb{R}_{> 0} \times S^2$ is compact without boundary. The radial derivative, on the other hand, yields a boundary term that vanishes as we impose an appropriate Dirichlet-type boundary condition on the Dirac spinors. More precisely, since the computation of the matrix $\Gamma^{\tau} \alpha^{r}$ leads to the expression
\begin{equation*}
\Gamma^{\tau} \alpha^{r} = \textnormal{i} \, \textnormal{diag} \biggl(- \frac{\Delta}{r_+}, r_+, r_+, - \frac{\Delta}{r_+}\biggr) \, ,
\end{equation*}
the radial boundary term becomes
\begin{equation*}
\iint_{S^2} \psi'^{\dagger} \, \Gamma^{\tau} \alpha^r \phi' \sin{(\theta)} \, \textnormal{d}\phi \, \textnormal{d}\theta \, \Big|^{r_2}_{r_1} = \textnormal{i} r_+ \iint_{S^2} \biggl(- \frac{\Delta}{r_+^2} \, \overline{\psi'}_1 \phi'_1 + \overline{\psi'}_2 \phi'_2 + \overline{\psi'}_3 \phi'_3 - \frac{\Delta}{r_+^2} \, \overline{\psi'}_4 \phi'_4\biggr) \sin{(\theta)} \, \textnormal{d}\phi \, \textnormal{d}\theta \, \Big|^{r_2}_{r_1} \, .
\end{equation*}
In order for this term to vanish, we impose the radial Dirichlet-type boundary condition
\begin{equation} \label{DTBCWI}
\sum_{i = 1}^2 \, (- 1)^i \biggl(- \frac{\Delta}{r_+^2} \, \overline{\psi'}_1 \phi'_1 + \overline{\psi'}_2 \phi'_2 + \overline{\psi'}_3 \phi'_3 - \frac{\Delta}{r_+^2} \, \overline{\psi'}_4 \phi'_4\biggr)_{|r = r_i} = 0 \, .
\end{equation}
In the present work, we consider only Dirac spinors with support from a specific time-like inner boundary at $r = r_0 < r_-$ beyond the Cauchy horizon up to infinity, that is
\begin{equation*}
\textnormal{supp} \, \phi' = (r_0, \infty) \times S^2 \, .
\end{equation*}
Moreover, we require the Dirac spinors to be in $L^2\bigl((r_0, \infty) \times S^2, SM\bigr)$, implying proper decay at infinity. Taking this into account, the radial boundary condition (\ref{DTBCWI}) reduces to a condition for the time-like inner boundary at $r = r_0$
\begin{equation} \label{DTBCWI2}
\biggl(- \frac{\Delta}{r_+^2} \, \overline{\psi'}_1 \phi'_1 + \overline{\psi'}_2 \phi'_2 + \overline{\psi'}_3 \phi'_3 - \frac{\Delta}{r_+^2} \, \overline{\psi'}_4 \phi'_4\biggr)_{|r = r_0} = 0 \, ,
\end{equation}
which can be brought into a more suitable form as follows. By means of the spin scalar product (\ref{spinsp}) and the relation $\mathscr{S}' \gamma'^r = \textnormal{i} \, \Gamma^{\tau} \alpha^{r}/\Sigma$, we may represent (\ref{DTBCWI2}) as
\begin{equation} \label{DirichletTBC2}
\Sl \psi' | \gamma'^r \phi' \Sr_{|\{\tau\} \times \{r_0\} \times S^2} = 0 \, .
\end{equation}
Now, introducing $\boldsymbol{n}$ as the unit normal to the hypersurfaces $\{\tau\} \times S^2$, we can write (\ref{DirichletTBC2}) in the form 
\begin{equation} \label{RadDBCFIN}
\Sl \psi' | \boldsymbol{\slashed{n}} \, \phi' \Sr_{|\{\tau\} \times \{r_0\} \times S^2} = 0 \quad \Leftarrow \quad (\boldsymbol{\slashed{n}} - \textnormal{i} \, \mathcal{H}) \, \psi'_{|\{\tau\} \times \{r_0\} \times S^2} = \boldsymbol{0} \, ,
\end{equation}
where the slash again denotes Clifford multiplication and $\mathcal{H}$ is an arbitrary matrix with the property $\mathcal{H} = (\mathscr{S}')^{- 1} \mathcal{H}^{\dagger} \mathscr{S}'$. The implication can be easily verified via the calculation
\begin{equation*} 
\begin{split}
\Sl \psi' | \boldsymbol{\slashed{n}} \, \phi' \Sr_{|\{\tau\} \times \{r_0\} \times S^2} & = \Sl \boldsymbol{\slashed{n}} \, \psi' | \phi' \Sr_{|\{\tau\} \times \{r_0\} \times S^2} = 
\Sl \textnormal{i} \, \mathcal{H} \, \psi' | \phi' \Sr_{|\{\tau\} \times \{r_0\} \times S^2} \\ 
& = - \Sl \psi' | \, \textnormal{i} \, \mathcal{H} \, \phi' \Sr_{|\{\tau\} \times \{r_0\} \times S^2} = - \Sl \psi' | \boldsymbol{\slashed{n}} \, \phi' \Sr_{|\{\tau\} \times \{r_0\} \times S^2} \, .
\end{split}
\end{equation*}
To guarantee compatibility of the boundary condition on the right hand side of (\ref{RadDBCFIN}) with a potential product structure of the Dirac $4$-spinors, in which the dependences on the variables $\tau, r, \theta$, and $\phi$ are separated (such as in Chandrasekhar's separation ansatz (\ref{Chandrasepans})), we choose 
\begin{equation*}
\mathcal{H} = \mathscr{P} (\mathscr{P}^{\dagger})^{- 1} \, ,
\end{equation*}
where $\mathscr{P}$ is defined in (\ref{spinortr}). We note in passing that this Dirichlet-type boundary condition is a so-called MIT-type boundary condition for Dirac fields \cite{CJJTW} that describes a perfect reflection of Dirac particles at the respective boundary surface. Continuing the proof of symmetry, the explicit computation of the square bracket in the fourth line of (\ref{PRSYM}) yields the result 
\begin{equation*} 
\partial_j (\Gamma^{\tau}) \, \alpha^j + \Gamma^{\tau} \partial_j (\alpha^j) + \Gamma^{\tau} \alpha^{\theta} \cot{(\theta)} = - 2 \mathscr{V}_0^{\dagger} \Gamma^{\tau} \, .
\end{equation*}
Besides, all three matrix products $\Gamma^{\tau} \alpha^j$, with $j \in \{r, \theta, \phi\}$, are anti-self-adjoint
\begin{equation*} 
\Gamma^{\tau} \alpha^j = - \alpha^{j \dagger} \Gamma^{\tau \dagger} = - \alpha^{j \dagger} \Gamma^{\tau} \, .
\end{equation*}
Therefore, we immediately find that
\begin{equation*}
\begin{split}
(\psi' | H \phi')_{|\mathfrak{N}_{\tau}} & = \iiint \partial_j (\psi'^{\dagger}) \, \alpha^{j \dagger} \, \Gamma^{\tau} \phi' \sin{(\theta)} \, \textnormal{d}\phi \, \textnormal{d}\theta \, \textnormal{d}r + 2 \iiint (\mathscr{V}_0 \psi')^{\dagger} \, \Gamma^{\tau} \phi' \sin{(\theta)} \, \textnormal{d}\phi \, \textnormal{d}\theta \, \textnormal{d}r \\ \\
& \hspace{0.4cm} - \iiint (\mathscr{V}_0 \psi')^{\dagger} \, \Gamma^{\tau} \phi' \sin{(\theta)} \, \textnormal{d}\phi \, \textnormal{d}\theta \, \textnormal{d}r + \iiint (\mathscr{V}_m \psi')^{\dagger} \, \Gamma^{\tau} \phi' \sin{(\theta)} \, \textnormal{d}\phi \, \textnormal{d}\theta \, \textnormal{d}r \\ \\
& = \iiint (\alpha^{j} \partial_j \psi')^{\dagger} \, \Gamma^{\tau} \phi' \sin{(\theta)} \, \textnormal{d}\phi \, \textnormal{d}\theta \, \textnormal{d}r + \iiint (\mathscr{V}_0 \psi')^{\dagger} \, \Gamma^{\tau} \phi' \sin{(\theta)} \, \textnormal{d}\phi \, \textnormal{d}\theta \, \textnormal{d}r \\ \\
& \hspace{0.4cm} + \iiint (\mathscr{V}_m \psi')^{\dagger} \, \Gamma^{\tau} \phi' \sin{(\theta)} \, \textnormal{d}\phi \, \textnormal{d}\theta \, \textnormal{d}r \\ \\
& = \iiint (H \psi')^{\dagger} \, \Gamma^{\tau} \phi' \sin{(\theta)} \, \textnormal{d}\phi \, \textnormal{d}\theta \, \textnormal{d}r = (H \psi' | \phi')_{|\mathfrak{N}_{\tau}} \, . 
\end{split}
\end{equation*}
\end{proof}

%------------------------------------------------------------------------------------------------
\section{Fundamental Solutions for the Construction of the Radial Green's Matrix} \label{appB}
%------------------------------------------------------------------------------------------------

\noindent In order to determine the fundamental solutions $\Phi_1(r; r')$ and $\Phi_2(r; r')$ of the radial system (\ref{RADODESYS}) with complex-valued frequencies, which are used for the construction of the Green's matrix defined via equation (\ref{Green}), we first study certain aspects of the associated Jost-type solutions \cite{ReedSimon3, alfaro+regge}. In more detail, we derive radial Jost-type equations that yield solutions with asymptotic behaviors near infinity, the event horizon, and the Cauchy horizon similar to those given in Lemmas \ref{L1} and \ref{L2a3}, and briefly discuss the existence, uniqueness, and boundedness of these solutions. Since we apply Lebesgue's dominated convergence theorem to simplify our integral spectral representation of the Dirac propagator, the latter aspect becomes also relevant for the commutation of specific limits, sums, and integrals. For the derivation of the Jost-type equations, we rewrite the radial first-order system (\ref{RADODESYS}) as two second-order scalar equations. In terms of the Regge--Wheeler coordinate $r_{\star}$ and the function $\widetilde{\mathscr{R}} = (\widetilde{\mathscr{R}}_+, \widetilde{\mathscr{R}}_-)^{\textnormal{T}} = \bigl(\sqrt{|\Delta|} \, \mathscr{R}_+, r_+ \, \mathscr{R}_-\bigr)^{\textnormal{T}}$, these read
\begin{equation} \label{radsystransf}
\bigl[\partial_{r_{\star} r_{\star}} + \mathfrak{J}_{\pm \xi, k, \omega}(r) \, \partial_{r_{\star}} + \mathfrak{K}^{\pm}_{\xi, k, \omega}(r)\bigr] \widetilde{\mathscr{R}}_{\pm} = 0 \, ,
\end{equation}
where
\begin{equation*} 
\begin{split}
\mathfrak{J}_{\xi, k, \omega}(r) & := \frac{1}{r^2 +a^2} \biggl[r - \frac{M (3 r^2 - a^2)}{r^2 + a^2} - 4 \textnormal{i} \omega M r - 2 \textnormal{i} k a - \frac{\textnormal{i} m \Delta}{\xi + \textnormal{i} m r}\biggl] \\ \\
\mathfrak{K}^+_{\xi, k, \omega}(r) & := \frac{\Delta}{(r^2 +a^2)^2} \biggl[\bigl(\textnormal{i} \omega \, [\Delta + 4 M r] + 2 \textnormal{i} k a\bigr)\biggl(- \textnormal{i} \omega + \frac{r - M}{\Delta} + \frac{\textnormal{i} m}{\xi + \textnormal{i} m r}\biggr) - 2 \textnormal{i} \omega \, (r + M) - m^2 r^2 - \xi^2\biggr] \\ \\
\mathfrak{K}^-_{\xi, k, \omega}(r) & := \frac{\Delta}{(r^2 +a^2)^2} \biggl[\textnormal{i} \omega \biggl(r - M - 4 \textnormal{i} \omega M r - 2 \textnormal{i} k a + \frac{\textnormal{i} m \Delta}{\xi - \textnormal{i} m r}\biggr) + \omega^2 \Delta - m^2 r^2 - \xi^2\biggr] \, .
\end{split}
\end{equation*}
Employing the ansatzes
\begin{equation*}
\widetilde{\mathscr{R}}_{\pm}(r_{\star}) = \exp{\biggl(- \frac{1}{2} \int \mathfrak{K}^{\pm}_{\xi, k, \omega}(r) \, \textnormal{d}r_{\star}\biggr)} \, \mathscr{Y}_{\pm}(r_{\star}) \, ,
\end{equation*}
we may transform (\ref{radsystransf}) into the Schr\"odinger-type equations
\begin{equation} \label{schrödinger}
\bigl[\partial_{r_{\star} r_{\star}} + \mathscr{V}^{\pm}_{\xi, k, \omega}(r)\bigr] \mathscr{Y}_{\pm} = 0
\end{equation}
with the potentials
\begin{equation*}
\mathscr{V}^{\pm}_{\xi, k, \omega}(r) := \mathfrak{K}^{\pm}_{\xi, k, \omega}(r) - \frac{\mathfrak{J}^2_{\pm \xi, k, \omega}(r)}{4} - \frac{\partial_{r_{\star}} \mathfrak{J}_{\pm \xi, k, \omega}(r)}{2} \, .
\end{equation*}
To obtain Jost-type equations with boundary conditions prescribed at infinity, we split these potentials into an asymptotic contribution effective at infinity and otherwise regular contributions
\begin{equation} \label{potsplit}
\mathscr{V}^{\pm}_{\xi, k, \omega} = \mathscr{V}_{\infty} + \mathscr{V}^{\pm}_{\textnormal{reg.}} \, ,
\end{equation}
where the asymptotic contribution is given by the expression
\begin{equation} \label{asymppot}
\mathscr{V}_{\infty} = \mathscr{V}_{\infty}(r_{\star}) := \omega^2 - m^2 + \frac{2 M m^2}{r_{\star}} 
\end{equation}
and the regular contributions are on the order of $\mathscr{V}^{\pm}_{\textnormal{reg.}} = \mathcal{O}\bigl(1/r_{\star}^2\bigr)$ satisfying the condition 
\begin{equation*}
\int_{r_{\star}}^{\infty} \big|\mathscr{V}^{\pm}_{\textnormal{reg.}}(y)\big| \, \textnormal{d}y < \infty \quad \textnormal{for all} \quad r_{\star} > 0 \, .
\end{equation*}
We remark that the asymptotic potential (\ref{asymppot}) corresponds to the equation
\begin{equation} \label{asympste}
\bigl[\partial_{r_{\star} r_{\star}} + \mathscr{V}_{\infty}(r_{\star})\bigr] \mathscr{Y}_{\infty} = 0 \, ,
\end{equation}
which has the solution \cite{WhittakerWatson}
\begin{equation*}
\mathscr{Y}_{\infty} = \mathcal{Z}_1 W_{- \alpha, \frac{1}{2}}\Bigl(2 \textnormal{i} \, \textnormal{sign}(\omega) \sqrt{\omega^2 - m^2} \, r_{\star}\Bigr) + \mathcal{Z}_2 W_{+ \alpha, \frac{1}{2}}\Bigl(- 2 \textnormal{i} \, \textnormal{sign}(\omega) \sqrt{\omega^2 - m^2} \, r_{\star}\Bigr) \, ,
\end{equation*}
where $W_{\pm \alpha, \frac{1}{2}}(\, \cdot \,)$ are Whittaker functions with $\alpha := \textnormal{i} \, \textnormal{sign}(\omega) \, M m^2/\sqrt{\omega^2 - m^2}$ and $\mathcal{Z}_{1/2}$ denote constants. The asymptotics of this solution at infinity reads in case $|\omega| \geq m$
\begin{equation} \label{asymsol1}
\begin{split}
\mathscr{Y}_{\infty} & \sim \mathcal{Z}'_1 \exp{\biggl(\textnormal{i} \, \textnormal{sign}(\omega) \biggl[\sqrt{\omega^2 - m^2} \, r_{\star} + \frac{M m^2}{\sqrt{\omega^2 - m^2}} \ln{(r_{\star})}\biggr]\biggr)} \\ \\
& \hspace{0.4cm} + \mathcal{Z}'_2 \exp{\biggl(- \textnormal{i} \, \textnormal{sign}(\omega) \biggl[\sqrt{\omega^2 - m^2} \, r_{\star} + \frac{M m^2}{\sqrt{\omega^2 - m^2}} \ln{(r_{\star})}\biggr]\biggr)} \, ,
\end{split}
\end{equation}
whereas for $|\omega| < m$ it yields
\begin{equation*} 
\begin{split}
\mathscr{Y}_{\infty} & \sim \mathcal{Z}'_1 \exp{\biggl(\textnormal{sign}(\omega) \biggl[\sqrt{m^2 - \omega^2} \, r_{\star} + \frac{M m^2}{\sqrt{m^2 - \omega^2}} \ln{(r_{\star})}\biggr]\biggr)} \\ \\
& \hspace{0.4cm} + \mathcal{Z}'_2 \exp{\biggl(- \textnormal{sign}(\omega) \biggl[\sqrt{m^2 - \omega^2} \, r_{\star} + \frac{M m^2}{\sqrt{m^2 - \omega^2}} \ln{(r_{\star})}\biggr]\biggr)} 
\end{split} 
\end{equation*}
with $\mathcal{Z}'_{1/2}$ also being constants (cf.\ Lemma \ref{L1}). In the following, we restrict our attention to the case $|\omega| \geq m$. The case $|\omega| < m$ may be treated accordingly. As in the usual study of Jost equations and their solutions, we complexify the Schr\"odinger-type equations (\ref{schrödinger}) via the analytic continuation $\omega \rightarrow \omega_{\textnormal{c}} \in \mathbb{C}$ of the frequency. Then, by means of the above splittings of the potentials (\ref{potsplit}) and the specific form of the asymptotic solution (\ref{asymsol1}), we can write the Jost-type equation representation of (\ref{schrödinger}) as
\begin{equation} \label{Josteq}
\begin{split}
\mathscr{Y}_{\pm}(r_{\star}) & = \exp{\biggl(\textnormal{i} \, \textnormal{sign}\bigl(\textnormal{Im}(\omega_{\textnormal{c}})\bigr) \, \textnormal{sign}(\omega_{\textnormal{c}}) \biggl[\sqrt{|\omega_{\textnormal{c}}|^2 - m^2} \,\, r_{\star} + \frac{M m^2}{\sqrt{|\omega_{\textnormal{c}}|^2 - m^2}} \ln{(r_{\star})}\biggr]\biggr)}	\\ \\
& \hspace{0.4cm} + \int_{r_{\star}}^{\infty} \frac{\sin{\bigl(\sqrt{\mathscr{V}_{\infty}(y)} \, [r_{\star} - y]\bigr)}}{\sqrt{\mathscr{V}_{\infty}(y)}} \, \mathscr{V}^{\pm}_{\textnormal{reg.}}(y) \, \mathscr{Y}_{\pm}(y) \, \textnormal{d}y \, .
\end{split}
\end{equation}
We note that the proper complexification of the asymptotic Schr\"odinger-type equation (\ref{asympste}), which is in accordance with the particular representation (\ref{asymsol1}) of the asymptotic solutions containing signum functions, is obtained by first rewriting the potential $\mathscr{V}_{\infty}$ defined in (\ref{asymppot}) in the form 
\begin{equation*} 
\mathscr{V}_{\infty} = \textnormal{sign}^2(\omega) \biggl(|\omega|^2 - m^2 + \frac{2 M m^2}{r_{\star}}\biggr)
\end{equation*}
and subsequently extending the frequency $\omega$ to complex values. This is relevant for the derivation of the exponential term in (\ref{Josteq}). Applying the series ansatzes
\begin{equation*}
\mathscr{Y}_{\pm}(r_{\star}) = \sum_{n = 0}^{\infty} \mathscr{Y}_{\pm, n}(r_{\star}) \, ,
\end{equation*}
where the zeroth-order terms are given by
\begin{equation*}
\mathscr{Y}_{\pm, 0}(r_{\star}) = \exp{\biggl(\textnormal{i} \, \textnormal{sign}\bigl(\textnormal{Im}(\omega_{\textnormal{c}})\bigr) \, \textnormal{sign}(\omega_{\textnormal{c}}) \biggl[\sqrt{|\omega_{\textnormal{c}}|^2 - m^2} \, r_{\star} + \frac{M m^2}{\sqrt{|\omega_{\textnormal{c}}|^2 - m^2}} \ln{(r_{\star})}\biggr]\biggr)} \, ,
\end{equation*}
in (\ref{Josteq}), we find the recurrence relations
\begin{equation*}
\mathscr{Y}_{\pm, n}(r_{\star}) = \int_{r_{\star}}^{\infty} \frac{\sin{\bigl(\sqrt{\mathscr{V}_{\infty}(y)} \, [r_{\star} - y]\bigr)}}{\sqrt{\mathscr{V}_{\infty}(y)}} \, \mathscr{V}^{\pm}_{\textnormal{reg.}}(y) \, \mathscr{Y}_{\pm, n - 1}(r_{\star}) \, \textnormal{d}y \quad \textnormal{for} \quad n \geq 1 \, .
\end{equation*}
In the theorem below, we discuss the relevant points pertaining to the existence, uniqueness, and boundedness of such solutions for the case $\textnormal{Im}{(\omega_{\textnormal{c}})} < 0$. Detailed proofs are worked out explicitly in, e.g., \cite{ReedSimon3, Kronthaler, FS, FKSY5}. The results and proofs for the case $\textnormal{Im}{(\omega_{\textnormal{c}})} > 0$ are in essence identical.
\begin{Thm} \label{TheoremJost}
For each $\omega_{\textnormal{c}} \in \mathbb{C}$ with $\omega_{\textnormal{c}} \not= 0$ and $\textnormal{Im}{(\omega_{\textnormal{c}})} < 0$, the Jost-type equations (\ref{Josteq}) have unique solutions $\mathscr{Y}_{\pm}(r_{\star})$ obeying
\begin{equation*}
\lim_{r_{\star} \rightarrow \infty} \bigg|\exp{\biggl(\textnormal{i} \, \textnormal{sign}(\omega_{\textnormal{c}}) \biggl[\sqrt{|\omega_{\textnormal{c}}|^2 - m^2} \, r_{\star} + \frac{M m^2}{\sqrt{|\omega_{\textnormal{c}}|^2 - m^2}} \ln{(r_{\star})}\biggr]\biggr)} \mathscr{Y}_{\pm}(r_{\star})\bigg| < \infty \, .
\end{equation*}
These solutions are moreover continuously differentiable in $r_{\star}$ on the interval $(0, \infty)$ with
\begin{equation*}
\lim_{r_{\star} \rightarrow \infty} \biggl[\exp{\biggl(\textnormal{i} \, \textnormal{sign}(\omega_{\textnormal{c}}) \biggl[\sqrt{|\omega_{\textnormal{c}}|^2 - m^2} \, r_{\star} + \frac{M m^2}{\sqrt{|\omega_{\textnormal{c}}|^2 - m^2}} \ln{(r_{\star})}\biggr]\biggr)} \mathscr{Y}_{\pm}(r_{\star})\biggr] = 1
\end{equation*}
and
\begin{equation*}
\lim_{r_{\star} \rightarrow \infty} \biggl[\exp{\biggl(\textnormal{i} \, \textnormal{sign}(\omega_{\textnormal{c}}) \biggl[\sqrt{|\omega_{\textnormal{c}}|^2 - m^2} \, r_{\star} + \frac{M m^2}{\sqrt{|\omega_{\textnormal{c}}|^2 - m^2}} \ln{(r_{\star})}\biggr]\biggr)} \partial_{r_{\star}}\mathscr{Y}_{\pm}(r_{\star})\biggr] = - \textnormal{i} \, \textnormal{sign}(\omega_{\textnormal{c}}) \sqrt{|\omega_{\textnormal{c}}|^2 - m^2} \, .
\end{equation*}
For each fixed value of $r_{\star}$, $\mathscr{Y}_{\pm}(r_{\star})$ and $\partial_{r_{\star}}\mathscr{Y}_{\pm}(r_{\star})$ are functions that are analytic in $\{\omega_{\textnormal{c}} \, | \, \textnormal{Im}{(\omega_{\textnormal{c}})} < 0\}$, continuous in $\{\omega_{\textnormal{c}} \, | \, \omega_{\textnormal{c}}  \not= 0 \,\,\, \textnormal{and} \,\,\, \textnormal{Im}{(\omega_{\textnormal{c}})} < 0\}$, and satisfy the bound
\begin{equation*}
\begin{split}
& \bigg|\mathscr{Y}_{\pm}(r_{\star}) - \exp{\biggl(- \textnormal{i} \, \textnormal{sign}(\omega_{\textnormal{c}}) \biggl[\sqrt{|\omega_{\textnormal{c}}|^2 - m^2} \, r_{\star} + \frac{M m^2}{\sqrt{|\omega_{\textnormal{c}}|^2 - m^2}} \ln{(r_{\star})}\biggr]\biggr)}\bigg| \\ \\
& \hspace{3.58cm}\leq \exp{\Bigl(\textnormal{Im}\bigl(\sqrt{\mathscr{V}_{\infty}(r_{\star})} \,\, \bigr) \, r_{\star}\Bigr)} \left|\exp{\bigl(\mathfrak{Q}^{\pm}(r_{\star})\bigr)} - 1\right| 
\end{split}
\end{equation*}
as well as
\begin{equation*}
\begin{split}
& \bigg|\partial_{r_{\star}} \mathscr{Y}_{\pm}(r_{\star}) + \exp{\biggl(- \textnormal{i} \, \textnormal{sign}(\omega_{\textnormal{c}}) \biggl[\sqrt{|\omega_{\textnormal{c}}|^2 - m^2} \, r_{\star} + \frac{M m^2}{\sqrt{|\omega_{\textnormal{c}}|^2 - m^2}} \ln{(r_{\star})}\biggr]\biggr)} \bigg. \\ \\
& \bigg. \times \, \textnormal{i} \, \textnormal{sign}(\omega_{\textnormal{c}}) \biggl(\sqrt{|\omega_{\textnormal{c}}|^2 - m^2} + \frac{M m^2}{\sqrt{|\omega_{\textnormal{c}}|^2 - m^2} \, r_{\star}}\biggr)\bigg| \leq \exp{\Bigl(\textnormal{Im}\bigl(\sqrt{\mathscr{V}_{\infty}(r_{\star})} \,\, \bigr) \, r_{\star} + \mathfrak{Q}^{\pm}(r_{\star})\Bigr)} \int_{r_{\star}}^{\infty} \big|\mathscr{V}^{\pm}_{\textnormal{reg.}}(y)\big| \, \textnormal{d}y \, ,
\end{split}
\end{equation*}
where
\begin{equation*}
\mathfrak{Q}^{\pm}(r_{\star}) := \int_{r_{\star}}^{\infty} \frac{4 y \, \big|\mathscr{V}^{\pm}_{\textnormal{reg.}}(y)\big|}{1 + y \, \big|\sqrt{\mathscr{V}_{\infty}(y)}\big|} \, \exp{\Bigl(\bigl[\textnormal{Im}\bigl(\sqrt{\mathscr{V}_{\infty}(y)} \,\, \bigr) + \big|\textnormal{Im}\bigl(\sqrt{\mathscr{V}_{\infty}(y)} \,\, \bigr)\big|\bigr] y\Bigr)} \, \textnormal{d}y \, .
\end{equation*}
\end{Thm}
\noindent 
It remains to determine the Jost-type equations with boundary conditions prescribed at the event horizon and at the Cauchy horizon. This can be done using a similar approach as in the above case. For details, we again refer to \cite{FS, FKSY5}.

We now specify the fundamental solutions $\Phi_1(r; r')$ and $\Phi_2(r; r')$ of the radial system (\ref{RADODESYS}) with complex-valued frequencies. Due to the high degree of complexity of this system, explicit analytical expressions for its fundamental solutions are not known. Thus, we describe them in terms of suitable asymptotic expansions. To this end, we define, on the one hand, auxiliary functions that have the proper decay at infinity (cf.\ Lemma \ref{L1})
\begin{equation*}
\begin{split}
& \widehat{\Phi}^{(\infty)}(r) := |\Delta|^{- 1/2} \biggl[d_{1, \infty} \, e^{\textnormal{i} \phi_+(r_{\star}(r))} \left(\begin{array}{cc} 1 \\ 0 \end{array}\right) + \mathcal{O}\biggl(\frac{1}{r_{\star}(r)}\biggr)\biggr] \hspace{0.48cm} \textnormal{for} \,\,\,\,\,\,\,\,  \begin{cases} \textnormal{Im}{(\omega_{\epsilon})} < 0 \,\,\,\,\, \textnormal{if} \,\,\,\,\, |\omega_{\epsilon}| \geq m \\ \textnormal{Re}{(\omega_{\epsilon})} \geq 0 \,\,\,\,\,\, \hspace{-0.045cm} \textnormal{if} \,\,\,\,\, |\omega_{\epsilon}| < m \end{cases} \\ \\
& \widecheck{\Phi}^{(\infty)}(r) := d_{2, \infty} \, e^{- \textnormal{i} \phi_-(r_{\star}(r))} \left(\begin{array}{cc} 0 \\ 1 \end{array}\right) + \mathcal{O}\biggl(\frac{1}{r_{\star}(r)}\biggr) \hspace{1.775cm} \textnormal{for} \,\,\,\,\,\,\,\, \begin{cases} \textnormal{Im}{(\omega_{\epsilon})} > 0 \,\,\,\,\, \textnormal{if} \,\,\,\,\, |\omega_{\epsilon}| \geq m \\ \textnormal{Re}{(\omega_{\epsilon})} < 0 \,\,\,\,\,\, \hspace{-0.045cm} \textnormal{if} \,\,\,\,\, |\omega_{\epsilon}| < m \, , \end{cases} 
\end{split}
\end{equation*}
where the quantities $d_{1/2, \infty}$ are scalar constants and the functions $\phi_{\pm}$ are given in (\ref{asypha}), but  with frequencies $\omega_{\epsilon} \in \{\omega + \textnormal{i} \epsilon, \omega - \textnormal{i} \epsilon\}$, for which $\epsilon > 0$, and with the substitution $\sqrt{\omega^2 - m^2} \rightarrow \sqrt{|\omega_{\epsilon}|^2 - m^2}$. On the other hand, we use auxiliary functions that are finite at the event and the Cauchy horizon and further comply with the associated asymptotics (cf.\ Lemma \ref{L2a3}) 
\begin{equation*}
\begin{split}
& \widehat{\Phi}^{(+)}(r) := |\Delta|^{- 1/2} \biggl[d_{1, r_+} \, e^{2 \textnormal{i} \bigl(\omega_{\epsilon} + k \Omega^{(+)}_{\textnormal{Kerr}}\bigr) r_{\star}(r)} \left(\begin{array}{cc} 1 \\ 0 \end{array}\right) + \mathcal{O}\bigl(e^{q r_{\star}(r)}\bigr)\biggr]  \hspace{0.6cm} \textnormal{for} \,\,\,\,\,\,\,\,  \textnormal{Im}{(\omega_{\epsilon})} < 0 \\ \\
& \widecheck{\Phi}^{(+)}(r) :=  d_{2, r_+} \left(\begin{array}{cc} 0 \\ 1 \end{array}\right) + \mathcal{O}\bigl(e^{q r_{\star}(r)}\bigr) \hspace{4.94cm} \textnormal{for} \,\,\,\,\,\,\,\,  \textnormal{Im}{(\omega_{\epsilon})} > 0 \\ \\
& \widehat{\Phi}^{(-)}(r) := d_{1, r_-} \left(\begin{array}{cc} 0 \\ 1 \end{array}\right) + \mathcal{O}\bigl(e^{- q r_{\star}(r)}\bigr) \hspace{4.715cm} \textnormal{for} \,\,\,\,\,\,\,\, \textnormal{Im}{(\omega_{\epsilon})} < 0 \\ \\
& \widecheck{\Phi}^{(-)}(r) := |\Delta|^{- 1/2} \biggl[d_{2, r_-} \, e^{2 \textnormal{i} \bigl(\omega_{\epsilon} + k \Omega^{(-)}_{\textnormal{Kerr}}\bigr) r_{\star}(r)} \left(\begin{array}{cc} 1 \\ 0 \end{array}\right) + \mathcal{O}\bigl(e^{- q r_{\star}(r)}\bigr)\biggr] \hspace{0.41cm} \textnormal{for} \,\,\,\,\,\,\,\, \textnormal{Im}{(\omega_{\epsilon})} > 0 \ ,
\end{split}
\end{equation*}
where $d_{1/2, r_{\pm}}$ are scalar constants as well. To clarify the notation, we point out that the superscripts $(\infty)$, $(+)$, and $(-)$ designate asymptotic expansions at infinity, the event horizon, and the Cauchy horizon, respectively. Besides, we remark that the existence and uniqueness of fundamental solutions of the radial system (\ref{RADODESYS}) with these particular asymptotic expansions follows from the above study of the radial Jost-type equations and, moreover, that these asymptotic expansions ensure that the fundamental solutions are square-integrable. For example, as the Regge--Wheeler coordinate $r_{\star}$ tends to minus infinity at the event horizon, the exponential factor in the auxiliary function $\widehat{\Phi}^{(+)}$ tends to zero because $\textnormal{Im}{(\omega_{\epsilon})} < 0$. However, this exponential factor would not be square-integrable if $\textnormal{Im}{(\omega_{\epsilon})} > 0$. Last, we introduce an auxiliary function that satisfies the Dirichlet-type boundary condition (\ref{DTBCKERRGEOM}) at $r = r_0$ 
\begin{equation*}
\Phi_{\partial M}(r) := \Phi_{\partial M}^{(1)}(r) \left(\begin{array}{cc} 1 \\ \textnormal{i} \sqrt{|\Delta|} \, / \, r_+ \end{array}\right) ,
\end{equation*}
with $\Phi_{\partial M}^{(1)}$ denoting its first component. Then, in case 
\begin{equation*}
|\omega_{\epsilon}| \geq m \,\,\,\,\, \textnormal{and} \,\,\,\,\, \textnormal{Im}{(\omega_{\epsilon})} < 0 \quad \textnormal{or} \quad |\omega_{\epsilon}| < m \, , \,\,\,\,\, \textnormal{Im}{(\omega_{\epsilon})} < 0 \, , \,\,\,\,\, \textnormal{and}  \,\,\,\,\, \textnormal{Re}{(\omega_{\epsilon})} \geq 0 \, ,
\end{equation*} 
the fundamental radial solutions $\Phi_1$ and $\Phi_2$ read \newline
\begin{equation} \label{FundSols1}
\begin{split}
& \Phi_1(r; r_+ < r' < \infty) = \Theta(r - r') \, \widehat{\Phi}^{(\infty)}(r) \\  
& \Phi_2(r; r_+ < r' < \infty) = \Theta(r' - r) \, \Theta(r - r_+) \, \widehat{\Phi}^{(+)}(r) \\ \\
& \Phi_1(r; r_- < r' \leq r_+) = \Theta(r - r') \bigl[\Theta(r_+ - r) \, \widehat{\Phi}^{(-)}(r) + \Theta(r - r_+) \, \widehat{\Phi}^{(\infty)}(r)\bigr] \\ 
& \Phi_2(r; r_- < r' \leq r_+) = \Theta(r - r') \, \Theta(r_+ - r) \, \widehat{\Phi}^{(+)}(r) \\ \\
& \Phi_1(r; r_0 \leq r' \leq r_-) = \Theta(r - r') \bigl[\Theta(r_+ - r) \, \widehat{\Phi}^{(-)}(r) + \Theta(r - r_+) \, \widehat{\Phi}^{(\infty)}(r)\bigr] \\ 
& \Phi_2(r; r_0 \leq r' \leq r_-) = \Theta(r' - r) \, \Phi_{\partial M}(r) \, ,
\end{split}
\end{equation}
whereas for
\begin{equation*}
|\omega_{\epsilon}| \geq m \,\,\,\,\, \textnormal{and} \,\,\,\,\, \textnormal{Im}(\omega_{\epsilon}) > 0 \quad \textnormal{or} \quad |\omega_{\epsilon}| < m \, , \,\,\,\,\, \textnormal{Im}(\omega_{\epsilon}) > 0 \, , \,\,\,\,\, \textnormal{and} \,\,\,\,\, \textnormal{Re}(\omega_{\epsilon}) < 0 \, ,
\end{equation*}
they are given by
\begin{equation}\label{FundSols2}
\begin{split}
& \Phi_1(r; r_+ < r' < \infty) = \Theta(r - r') \, \widecheck{\Phi}^{(\infty)}(r) \\ 
& \Phi_2(r; r_+ < r' < \infty) = \Theta(r' - r) \bigl[\Theta(r - r_-) \, \widecheck{\Phi}^{(+)}(r) + \Theta(r_- - r) \, \Phi_{\partial M}(r)\bigr] \\ \\
& \Phi_1(r; r_- < r' \leq r_+) = \Theta(r' - r) \bigl[\Theta(r - r_-) \, \widecheck{\Phi}^{(+)}(r) + \Theta(r_- - r) \, \Phi_{\partial M}(r)\bigr] \\ 
& \Phi_2(r; r_- < r' \leq r_+) = \Theta(r' - r) \, \Theta(r - r_-) \, \widecheck{\Phi}^{(-)}(r) \\ \\
& \Phi_1(r; r_0 \leq r' \leq r_-) = \Theta(r - r') \, \Theta(r_- - r) \, \widecheck{\Phi}^{(-)}(r) \\ 
& \Phi_2(r; r_0 \leq r' \leq r_-) = \Theta(r' - r) \, \Phi_{\partial M}(r) \, .
\end{split}
\end{equation}
In the remaining cases 
\begin{equation*}
|\omega_{\epsilon}| < m \, , \,\,\,\,\, \textnormal{Im}(\omega_{\epsilon}) < 0 \, , \,\,\,\,\, \textnormal{and} \,\,\,\,\, \textnormal{Re}(\omega_{\epsilon}) < 0 \quad \textnormal{or} \quad |\omega_{\epsilon}| < m \, , \,\,\,\,\, \textnormal{Im}{(\omega_{\epsilon})} > 0 \, , \,\,\,\,\, \textnormal{and} \,\,\,\,\, \textnormal{Re}{(\omega_{\epsilon})} \geq 0 \, ,
\end{equation*}
we also obtain the fundamental solutions (\ref{FundSols1}) or (\ref{FundSols2}), respectively, but with the auxiliary functions $\widehat{\Phi}^{(\infty)}$ and $\widecheck{\Phi}^{(\infty)}$ interchanged. A case-by-case analysis shows that these solutions are uniquely determined by the conditions and asymptotics listed in the proof of Theorem \ref{Prop}. 

\end{appendix}

\vspace{0.5cm}

%------------------------------------------

\end{document}